\documentclass{article}
\usepackage{PRIMEarxiv}

\usepackage[utf8]{inputenc} 
\usepackage[T1]{fontenc}    
\usepackage[colorlinks=true,citecolor=blue,linkcolor=blue]{hyperref}
\usepackage{url}            
\usepackage{nicefrac}       
\usepackage{microtype}      
\usepackage{lipsum}
\usepackage{fancyhdr}       
\usepackage{graphicx}       
\graphicspath{{media/}}     

\usepackage{amsmath,amsthm,amssymb,amsfonts,mathtools}
\usepackage{booktabs,array,tabularx}
\usepackage{algorithm,algorithmic}
\usepackage{enumitem}
\usepackage[sort,numbers]{natbib}
\usepackage{xcolor}         

\usepackage{libertine}

\newtheorem{theorem}{Theorem}
\newtheorem*{theorem*}{Theorem}

\newcommand{\suppmat}{Appendix}

\newcommand{\proofmain}[1]{\vspace{-0.0mm} \noindent \textbf{Proof:} See \suppmat{} \ref{prf:#1}. \vspace{1.2mm} \newline}

\newcommand{\parbasic}[1]{\noindent \textbf{#1} \hspace{0.3mm}}

\newcommand{\Schrodinger}{Schr\"{o}dinger}
\newcommand{\ansatze}{ans\"{a}tze}

\DeclareMathOperator*{\argmin}{argmin}

\newcommand{\na}{{\mathcal{N}_\alpha}}

\newcommand{\elocal}{\mathcal{E}}

\newcommand{\G}{\mathbb{G}}
\newcommand{\nup}{\mathcal{N}_u}
\newcommand{\ndn}{\mathcal{N}_d}
\renewcommand{\l}{\ell}
\newcommand{\lp}{{\l + 1}}

\newcommand{\alphac}{{\hat{\alpha}}}

\newcommand{\ral}{r_\alpha^\l}
\newcommand{\rcl}{r_\alphac^\l}
\newcommand{\rail}{r_{\alpha,i}^\l}

\newcommand{\railp}{r_{\alpha,i}^\lp}

\newcommand{\bal}{\beta_\alpha^\l}
\newcommand{\balp}{\beta_\alpha^\lp}
\newcommand{\gal}{\gamma_\alpha^\l}
\newcommand{\gcl}{\gamma_\alphac^\l}
\newcommand{\gail}{\gamma_{\alpha,i}^\l}
\newcommand{\gailp}{\gamma_{\alpha,i}^\lp}

\newcommand{\numa}{{n_\alpha}}

\newcommand{\xal}{\xi_\alpha^\l}

\newcommand{\gailrp}{\underline{\gamma}_{\alpha,i}^\lp}
\newcommand{\galrp}{\underline{\gamma}_\alpha^\lp}
\newcommand{\gclrp}{\underline{\gamma}_\alphac^\lp}

\newcommand{\Pal}{\varphi_\alpha^\l}
\newcommand{\Pail}{\varphi_{\alpha,i}^\l}
\newcommand{\zal}{\zeta_\alpha^\l}
\newcommand{\zail}{\zeta_{\alpha,i}^\l}
\newcommand{\bzal}{\bar{\zeta}_\alpha^\l}
\newcommand{\bzcl}{\bar{\zeta}_\alphac^\l}

\newcommand{\Tal}{T_\alpha^\l}
\newcommand{\Tail}{T_{\alpha,i}^\l}
\newcommand{\tal}{\tau_\alpha^\l}
\newcommand{\mal}{\mu_\alpha^\l}
\newcommand{\hmal}{\hat{\mu}_\alpha^\l}
\newcommand{\Lal}{\Lambda_\alpha^\l}

\newcommand{\pir}{{\pi_\prec(r)}}

\newcommand{\ril}{r_i^\l}
\newcommand{\rilp}{r_i^\lp}
\newcommand{\rjl}{r_j^\l}
\newcommand{\rjlp}{r_j^\lp}
\newcommand{\Dijl}{\delta_{ij}^\l}
\newcommand{\Dijlp}{\delta_{ij}^\lp}
\newcommand{\Dl}{\delta^\l}
\newcommand{\Dlp}{\delta^\lp}
\newcommand{\DiIl}{\delta_{iI}^\l}
\newcommand{\DIl}{\delta_I^\l}

\newcommand{\Rl}{\Theta^\l}
\newcommand{\tl}{t^\l}
\newcommand{\zijl}{\zeta_{ij}^\l}
\newcommand{\bzl}{\bar{\zeta}^\l}
\newcommand{\el}{\eta^\l}
\newcommand{\hel}{\hat{\eta}^\l}

\newcommand{\dpp}{\rho_{dpp}}

\title{
\Large
A Theoretical Framework for an Efficient Normalizing Flow-Based Solution to the Electronic \Schrodinger{} Equation
}

\author{
  Daniel Freedman\thanks{Independent Researcher}
  \And
  Eyal Rozenberg\footnotemark[1]\hspace{1.4mm}\footnotemark[2]
  \And
  Alex Bronstein\thanks{Technion - Israel Institute of Technology}
}

\begin{document}

\maketitle

\begin{abstract}
    A central problem in quantum mechanics involves solving the Electronic \Schrodinger{} Equation for a molecule or material. The Variational Monte Carlo approach to this problem approximates a particular variational objective via sampling, and then optimizes this approximated objective over a chosen parameterized family of wavefunctions, known as the ansatz.  Recently neural networks have been used as the ansatz, with accompanying success.  However, sampling from such wavefunctions has required the use of a Markov Chain Monte Carlo approach, which is inherently inefficient.  In this work, we propose a solution to this problem via an ansatz which is cheap to sample from, yet satisfies the requisite quantum mechanical properties.  We prove that a normalizing flow using the following two essential ingredients satisfies our requirements: (a) a base distribution which is constructed from Determinantal Point Processes; (b) flow layers which are equivariant to a particular subgroup of the permutation group.  We then show how to construct both continuous and discrete normalizing flows which satisfy the requisite equivariance.  We further demonstrate the manner in which the non-smooth nature (``cusps'') of the wavefunction may be captured, and how the framework may be generalized to provide induction across multiple molecules.  The resulting theoretical framework entails an efficient approach to solving the Electronic \Schrodinger{} Equation.
\end{abstract}

\section{Introduction}
\label{sec:introduction}

\parbasic{The Electronic \Schrodinger{} Equation}
A central problem in quantum mechanics involves solving the Electronic \Schrodinger{} Equation to compute the ground state energy and wavefunction of a molecule or material.  This problem has manifold applications in chemistry, condensed matter physics, and materials science. However, solving the Electronic \Schrodinger{} Equation can be tricky: analytical solutions only exist for certain very simple problems; and numerical solutions must cope with the curse of dimensionality, due to the fact that the wavefunction's dimensionality grows linearly with the number of electrons.

\parbasic{Variational Monte Carlo}
A standard computational approach to this problem is based on Variational Monte Carlo \citep{ceperley1986quantum, austin2012quantum, gubernatis2016quantum, foulkes2001quantum, needs2009continuum}.  Solving for the ground state is equivalent to finding the eigenfunction of the Hamiltonian corresponding to the smallest eigenvalue (energy); this problem can be posed as the minimization of a Rayleigh quotient.  The issue is that both the numerator and denominator of the Rayleigh quotient correspond to very high dimensional integrals.  Variational Monte Carlo approximates these integrals via sampling, and the approximated objective is optimized over a family of wavefunctions, yielding an upper bound on the ground state energy.  The heart of this method is therefore the wavefunction family, also known as the ansatz: the more expressive the ansatz, the tighter the upper bound will be, yielding better approximations to the ground state energy.

\parbasic{Neural Networks as the Ansatz}
Recent work has proposed using neural networks as a flexible ansatz, and has achieved very high quality results.  The network must be adapted to respect the properties of electronic wavefunctions -- particularly antisymmetry, as electrons are Fermions.  Important examples of this type of ansatz are PauliNet \citep{hermann2020deep, schatzle2021convergence} and FermiNet \citep{pfau2020ab, spencer2020better}, both of which attain excellent ground state energies (e.g.~FermiNet achieves 99.8\% of the correlation energy for boron atoms).

\parbasic{The Problem: Sampling Inefficiency}
In order to be able to apply the Variational Monte Carlo formalism to the \ansatze{} just described, such as PauliNet or FermiNet, one must be able to sample from the densities corresponding to the wavefunctions given by their neural networks.  In general, this is only possibly using Markov Chain Monte Carlo (MCMC) techniques such as Langevin Monte Carlo \citep{umrigar1993diffusion} or any of several variations.  The issue with using such MCMC approaches to sampling is that they are inherently time-consuming: each sample is itself the solution of a stochastic differential equation as time goes to infinity.

\parbasic{Goals and Contributions}
The main goal of this paper is to solve the problem of sampling inefficiency, thereby yielding faster algorithms for solving the Electronic \Schrodinger{} Equation.  We achieve this goal by specifying a wavefunction ansatz which is easy to sample from, yet satisfies the requisite quantum mechanical properties.  In particular, we use an ansatz based on specially designed normalizing flows.  More specifically, we provide the following contributions:
\begin{itemize}
    \item We establish that an ansatz satisfying our desiderata can be instantiated as a normalizing flow with these characteristics: (a) its base distribution is symmetric under a particular subgroup of the permutation group, and vanishes for identical electrons; (b) the flow transformation is equivariant to the same subgroup of the permutation group.
    \item We show that the base distribution can be constructed using a particular combination of Determinantal Point Processes.
    \item We construct both continuous and discrete normalizing flows obeying the requisite equivariance.
    \item We provide a training regimen based on standard stochastic gradient descent.
    \item We show how to accommodate cusps, which encapsulate non-smooth aspects of the wavefunction.
    \item We generalize the framework so that induction across multiple molecules may be accommodated, while including the necessary additional invariances, in particular rigid motion invariance.
\end{itemize}
We end by noting that our contributions are of a theoretical character.  Our interest is to elucidate a theoretical approach to solving the Electronic \Schrodinger{} Equation based on neural networks, which retains both the expressivity of the neural network function class while at the same time providing considerably greater efficiency.  These contributions should be of importance and interest to the growing community of researchers working on the boundary of machine learning and quantum physics; including those researchers whose work involves the practical computation of ground state energies.

\section{Related Work}
\label{sec:related_work}

\parbasic{Neural Networks and Physics}
We begin by noting that neural networks have recently found broad and quite varied use in physics problems.  Amongst others, examples include solution of physics-based partial differential equations \citep{raissi2019physics, li2020fourier, lu2021learning}; density functional theory \citep{ryczko2019deep, kirkpatrick2021pushing}; topological states \citep{deng2017machine}; quantum state tomography \citep{torlai2018neural}; and quantum optics \citep{rozenberg2021inverse, rozenberg2022inverse}.  Useful surveys include \citep{carleo2019machine, zhang2023artificial}.

\parbasic{Pure Spin Systems}
Various works have used neural networks as the ansatz in the case of pure-spin systems, sometimes also referred to as ``discrete space systems'', in which the spins are at fixed locations.  These include \citep{carleo2017solving, deng2017quantum, gao2017efficient, levine2019quantum, sharir2022neural, passetti2023can}.

\parbasic{Continuous Space Systems}
In terms of continuous space problems of the sort that interest us, DeepWF \citep{han2019solving} bases its on ansatz on the classical Slater-Jastrow formalism, but learns both the symmetric and antisymmetric parts; the latter contains only two-electron terms, limiting the accuracy.  PauliNet \citep{hermann2020deep, schatzle2021convergence} also bases its ansatz on the Slater-Jastrow-Backflow form, but does so in a way that captures many-electron interactions, while respecting permutation-equivariance; this, as well as the inclusion of cusp terms, leads to much higher accuracy (e.g.~97.3\% of the correlation energy for boron atoms).  FermiNet \citep{pfau2020ab, spencer2020better} attains still higher accuracy (e.g.~99.8\% of the correlation energy for boron atoms) by using an appropriately designed neural network to represent the entire wavefunction, which contains a generalization of Slater determinants to account for all-electron interactions.  A hybrid solution which improves upon both PauliNet and FermiNet is presented in \citep{gerard2022gold}.  Techniques for learning / induction across several molecules or materials at once are presented in \citep{gao2023generalizing, scherbela2024towards, gerard2024transferable}.  We briefly mention applications to periodic systems \citep{wilson2022wave, li2022ab, pescia2022neural, cassella2023discovering}; techniques that use Diffusion Monte Carlo \citep{wilson2021simulations, ren2023towards}; and methods that deal with excited states \citep{entwistle2023electronic, pfau2023natural, naito2023multi}.

\parbasic{Normalizing Flows and Quantum Problems}
There is a body of work which has applied normalizing flows to problems in lattice field theories, e.g. \citep{kanwar2020equivariant, boyda2021sampling, abbott2023aspects, abbott2023normalizing}; due to the different physical scenario, the setup in these papers is naturally quite different from ours.  The work by \citep{xie2022abinitio} introduced the use of normalizing flows in the context of Fermionic problems; follow on papers for specific applications include \citep{xie2023m, saleh2023computing}.  Our paper considerably broadens this approach, by (a) establishing all results rigorously, as in Theorems \ref{thm:density_phase} - \ref{thm:continuous_flow}; (b) showing how to practically implement the flows with discrete layers, as in Theorem \ref{thm:layer_inverse_equivariance}; (c) providing important extensions to deal with phase computation (see Theorem \ref{thm:phase}), the crucial issue of cusps (see Theorem \ref{thm:cusps}), and induction across multiple molecules (see Theorems \ref{thm:induction} - \ref{thm:induction_transformation}).  Finally, we mention two further papers: \citep{stokes2023numerical} focuses on geometric aspects of the problem; while  WaveFlow \citep{thiede2022waveflow}, develops a specialized normalizing flow architecture for one-dimensional systems, which is interesting, albeit of limited applicability.

\section{Problem Setup}
\label{sec:problem_setup}

\subsection{Goals}

\parbasic{The Setting}
Our overall goal is to compute the ground state wavefunction and energy of a molecule given its molecular parameters and spin multiplicity.  Denote $x_i = (r_i, s_i)$ to be the pair consisting of the position and spin for the $i^{th}$ electron; $x$ will denote the entire ordered list $(x_1, \dots, x_n)$, with corresponding definitions for $r$ and $s$.  We specify wavefunctions as $\psi(x)$; due to the fact that electrons are Fermions, valid wavefunctions must be \textit{antisymmetric}, that is if $\pi \in \mathbb{S}_n$ is a permutation, then
\begin{equation}
    \psi(\pi x) = (-1)^\pi \psi(x)
\end{equation}
where as usual, $(-1)^\pi$ is shorthand for $(-1)^{N(\pi)}$ where $N(\pi)$ is the minimal number of flips to produce $\pi$.

Let $R_I$ and $Z_I$ denote the position and atomic number of the $I^{th}$ nucleus, and
let the Laplacian for the $i^{th}$ electron be $\Delta_i = \frac{\partial^2}{\partial r_{i1}^2} + \frac{\partial^2}{\partial r_{i2}^2} + \frac{\partial^2}{\partial r_{i3}^2}$, and denote the sum of the Laplacians by $\Delta = \sum_i \Delta_i$; then the Hamiltonian is given by
\begin{equation}
    H = -\tfrac{1}{2}\Delta + V(x)
\end{equation}
where the potential $V(x)$ is given by
\begin{equation}
    \textstyle
    V(x) = \sum_{i>j} \frac{1}{\| r_i - r_j \|} - \sum_{iI}  \frac{Z_I}{\| r_i - R_I \|} + \sum_{I>J}  \frac{Z_I Z_J}{\| R_I - R_J \|}
\end{equation}
Our goal is to compute the ground state wavefunction, which we denote as $\psi_0(x)$ and corresponding ground state energy $E_0$.  They may be computed using the variational principle, i.e.~by minimizing the Rayleigh quotient:
\begin{equation}
    \psi_0 = \argmin_{\psi \in \Psi} \frac{\langle \psi | H | \psi \rangle}{\langle \psi | \psi \rangle}
    \quad \text{and} \quad
    E_0 = \frac{\langle \psi_0 | H | \psi_0 \rangle}{\langle \psi_0 | \psi_0 \rangle}
    \label{eq:psi_argmin}
\end{equation}
where $\Psi$ is the set of all possible valid wavefunctions, and $H$ is the Hamiltonian.  If we specify the wavefunction ansatz as a neural network with parameters $\theta$, this becomes
\begin{equation}
    \theta^* = \argmin_\theta \frac{\langle \psi(\cdot; \theta) | H | \psi(\cdot; \theta) \rangle}{\langle \psi(\cdot; \theta) | \psi(\cdot; \theta) \rangle}
\end{equation}
and
\begin{equation}
    E^* = \frac{\langle \psi(\cdot; \theta^*) | H | \psi(\cdot; \theta^*) \rangle}{\langle \psi(\cdot; \theta^*) | \psi(\cdot; \theta^*) \rangle} \ge E_0
    \label{eq:loss}
\end{equation}
That is, we compute an upper bound $E^*$ to the ground state energy $E_0$.  The more expressive the ansatz, the tighter the bound will be.

\parbasic{Variational Monte Carlo}
The issue with the formulation to this point is the need to compute the inner products in Equations (\ref{eq:psi_argmin}) and (\ref{eq:loss}), which correspond to very high-dimensional integrals.  A standard solution to this problem is based on a Monte Carlo scheme.  To begin with, let us define the local energy $\elocal(x)$ and its real part $\elocal_r(x)$ as
\begin{equation}
    \elocal(x) \equiv \frac{H \psi(x)}{\psi(x)} = -\frac{\Delta \psi(x)}{2\psi(x)} + V(x)
    , \hspace{2mm}
    \elocal_r(x) = \mathfrak{Re}\{ \elocal(x) \}
    \label{eq:elocal_original}
\end{equation}
In this case, one can simplify the minimand in Equation (\ref{eq:psi_argmin}) (see \suppmat{}) as
\begin{equation}
    \frac{\langle \psi | H | \psi \rangle}{\langle \psi | \psi \rangle}
    \, = \, \mathbb{E}_{x \sim \rho(\cdot)} \left[ \elocal_r(x) \right]
    \, \approx \, \frac{1}{K} \sum_{k=1}^K \elocal_r\left(x^{(k)}\right)
    \label{eq:loss_for_sgd}
\end{equation}
where the $x^{(k)}$ are sampled from
\begin{equation}
    \rho(x) = \frac{|\psi(x)|^2}{\langle \psi | \psi \rangle}.
\end{equation}

\subsection{The General Approach}

As detailed in Sections \ref{sec:introduction} and \ref{sec:related_work}, a number of recent works have followed the above approach using a variety of neural networks as the ansatz for the wavefunction $\psi(\cdot; \theta)$.  In order to do so, one must be able to sample from $\rho(x; \theta) = |\psi(x; \theta)|^2 / \langle \psi | \psi \rangle$; and as the networks are quite general, the only feasible method for sampling is a Markov Chain Monte Carlo technique such as Langevin Monte Carlo \citep{umrigar1993diffusion} or any of several variations.  This can be time-consuming, as each sample is the solution of a stochastic differential equation as time goes to infinity.

A solution to this problem presents itself if we can somehow specify a wavefunction $\psi(x)$ which is easy to sample from.  We are interested in wavefunctions which satisfy the following three properties:
\begin{enumerate}[label=(W\arabic*)]
    \item There is an explicit functional form for the wavefunction $\psi(x)$.
    \item $\psi$ is antisymmetric.
    \item We can sample non-iteratively (in constant time) from $|\psi(\cdot)|^2$.
\end{enumerate}
The first two properties are necessary for any form of Variational Monte Carlo: (W1) allows us to evaluate the local energy $\elocal_r$ in (\ref{eq:elocal_original}) for use in (\ref{eq:loss_for_sgd}); and (W2) is required for valid electronic (Fermionic) wavefunctions.  But (W3) is the new ingredient: if we have a family of wavefunctions $\psi$ satisfying (W1)-(W3), then solving the minimization in (\ref{eq:loss}) via the Monte Carlo approach in (\ref{eq:loss_for_sgd}) will be considerably accelerated, as each sample will only require constant time to generate.  We add a fourth property, which is not strictly necessary but is both desirable and will prove useful:
\begin{enumerate}[label=(W\arabic*)]
    \setcounter{enumi}{3}
    \item $\psi$ is normalized, that is $\int |\psi(x)|^2 dx = 1$.
\end{enumerate}

It turns out that generating such wavefunctions is possible using the following procedure:
\begin{theorem}
    \label{thm:density_phase}
    Let $\rho(\cdot)$ be a probability density function which we can sample from in constant time.  Let $\rho(\cdot)$ satisfy two additional properties:
    \begin{enumerate}[label=(D\arabic*)]
        \item $\rho(x)$ is symmetric: $\rho(\pi x) = \rho(x)$ for all permutations $\pi \in \mathbb{S}_n$.
        \item $\rho(x) = 0$ if $x_i = x_j$ for any $i, j$.
    \end{enumerate}
    Finally, let $\kappa(x)$ be a complex function which satisfies $|\kappa(x)| = 1 \,\, \forall x$, and is nearly antisymmetric: 
    \begin{equation}
        \kappa(\pi x) = 
        \begin{cases}
            (-1)^\pi \kappa(x) & \text{if } x_i \neq x_j \text{ for all } i, j \\
            \bar{\kappa} & \text{otherwise}
        \end{cases}
    \end{equation}
    where $\bar{\kappa} \in \mathbb{C}$ is an arbitrary value with $|\bar{\kappa}|=1$.
    Then $\psi$ satisfies (W1)-(W4) if and only if $\psi$ can be written as $\psi(x) = \kappa(x) \sqrt{\rho(x)}$ with $\kappa$ and $\rho$ satisfying the above-stated properties.
\end{theorem}
\proofmain{density_phase}
The general idea expressed in Theorem \ref{thm:density_phase} is that we can build the wavefunction $\psi$ out of an easy-to-sample-from density function satisfying additional properties (D1)-(D2); and a nearly antisymmetric phase function $\kappa$.  In what follows, we will show how to construct both of these ingredients.  But before doing so, we take a short detour to address the most important practical scenario, that of fixed spin multiplicity.

\subsection{Fixed Spin Multiplicity}
\label{sec:multiplicity}

\parbasic{Notation}
As in most approaches to this problem, we assume that the spin multiplicity of the molecule is specified, which is equivalent to fixing the number of spin up and spin down electrons, denoted $n_u$ and $n_d$ respectively, with $n_u + n_d = n$.  Define the \textit{canonical spin vector} to be given by $\bar{s} = [\uparrow, \dots, \uparrow, \downarrow, \dots, \downarrow]$, i.e.~the first $n_u$ are $\uparrow$, the last $n_d$ are $\downarrow$.
We let the sets of indices of up and down spin electrons for the canonical spin vector be denoted by $\nup = \{ 1, \dots, n_u \}$ and $\ndn = \{ n_u + 1, \dots, n \}$.  Finally, we will be interested in the subgroup of permutations in which a permutation is applied separately to spin-up and spin-down electrons.  We denote this subgroup by
\begin{equation}
    \G \equiv \mathbb{S}_{\nup} \times \mathbb{S}_{\ndn} \quad \quad (\G \text{ is a subgroup of } \mathbb{S}_n)
\end{equation}

\parbasic{Specification of the Density}
In the case of fixed spin multiplicity, the specification of the density $\rho(x)$ is simplified:
\begin{theorem}
    \label{thm:spin_multiplicity}
    Given a configuration $x = (r, s)$, let a permutation which maps the spin vector $s$ to the canonical spin vector $\bar{s}$ be given by $\bar{\pi}_s$, i.e. $\bar{s} = \bar{\pi}_s s$.  Let $\bar{\rho}(r)$ be a density function on electron positions (i.e.~no spins) satisfying
    \begin{enumerate}[label=(R\arabic*)]
        \item $\bar{\rho}$ is $\G$-invariant: $\bar{\rho}(\pi r) = \bar{\rho}(r) \text{ for all } \pi \in \G$
        \item $\bar{\rho}(r) = 0 \text{ if } r_i = r_j, \text{ for } i, j \in \nup \text{ or } i, j \in \ndn$
    \end{enumerate}
    A density $\rho(x) = \rho(r, s)$ satisfies conditions (D1)-(D2) in Theorem \ref{thm:density_phase} if and only if it may be written as  $\rho(r, s) = \bar{\rho}(\bar{\pi}_s r)$ for a density  $\bar{\rho}(r)$ satisfying conditions (R1) and (R2).
\end{theorem}
\proofmain{spin_multiplicity}
To summarize: in the case of fixed spin multiplicity, specifying a wavefunction $\psi$ satisfying our desired conditions (W1)-(W4) is equivalent to specifying a density $\bar{\rho}(r)$ satisfying conditions (R1)-(R2); and then applying the transformations given in Theorems \ref{thm:density_phase} and \ref{thm:spin_multiplicity} to map from $\bar{\rho}$ to $\psi$.\footnote{We have for the moment ignored the issue of the phase $\kappa$, which we return to in Sections \ref{sec:training_via_sgd} and \ref{sec:phase}.}
Therefore, henceforth we will focus exclusively on specifying densities $\bar{\rho}(r)$ satisfying conditions (R1)-(R2).  To avoid unnecessary notational complexity, we will drop the bars and simply write $\rho(r)$.

\section{Using Normalizing Flows to Construct the Wavefunction Ansatz}
\label{sec:normalizing_flow}

\subsection{Sufficient Properties of the Normalizing Flow's Base Density and Transformation}

Our goal is to use a normalizing flow to construct the density $\rho(r)$.  Let $D$ be the ambient dimension (i.e.~$D=3$) and $n$ be the number of electrons.  The relevant vectors will live in the space $\mathbb{R}^{Dn}$ construed as the Cartesian product $\mathbb{R}^D \times \dots \times \mathbb{R}^D$ (which is of course isomorphic to $\mathbb{R}^{Dn}$).  A normalizing flow will consist of two ingredients: (1) a base random variable $z$, which lives in $\mathbb{R}^{Dn}$, and is described by the density $\rho_z(z)$; (2) an invertible transformation $T: \mathbb{R}^{Dn} \to \mathbb{R}^{Dn}$, such that $r = T(z)$.  In this case, the density $\rho(r)$ is the push-forward of $\rho_z$ along $T$, and is given by the change of variables formula
\begin{equation}
    \rho(r) = \rho_z(T^{-1}(r))|\det J_{T^{-1}}(r)|
    \label{eq:change_of_variables}
\end{equation}
Recall that we would like our density $\rho(r)$ to satisfy conditions (R1)-(R2) laid out in Theorem \ref{thm:spin_multiplicity}.  The following theorem establishes conditions for this to occur:
\begin{theorem}
    \label{thm:normalizing_flow}
    Suppose that we have a normalizing flow, whose base density $\rho_z$ satisfies properties (R1) and (R2) from Theorem \ref{thm:spin_multiplicity}, and whose transformation $T$ is $\G$-equivariant.  Then the density resulting from the normalizing flow will satisfy properties (R1) and (R2).
\end{theorem}
\proofmain{normalizing_flow}
Armed with this key result, we now set out to design the base density $\rho_z$ and transformation $T$ which satisfy the conditions of Theorem \ref{thm:normalizing_flow}.

\subsection{Base Density: Determinantal Point Processes}

In most cases in machine learning, the base density for a normalizing flow is taken to be a standard distribution, most often a Gaussian.  In our case, we require that the base density have certain special properties, namely (R1) and (R2) from Theorem \ref{thm:spin_multiplicity}.  It turns out that Determinantal Point Processes (DPPs) have just the properties we require.  In particular, we are interested in the class of DPPs known as Projection DPPs \citep{gautier2019two, lavancier2015determinantal}, which can be specified as follows.  We will let $y$ specify a generic point in $\mathbb{R}^D$. Let $h_k : \mathbb{R}^D \to \mathbb{R}$ for $k = 1, \dots, n$ be a set of $n$ functions which are orthogonal, that is $\langle h_i, h_j \rangle = \int_{\mathbb{R}^D} h_i(y) h_j(y) dy = \delta_{ij}$.  Let $H(y)$ be the column vector composed by stacking the individual functions $h_i(y)$ and define the kernel function as $K(y, y') = H(y)^T H(y')$.  Then for a given collection of $n$ points in $\mathbb{R}^D$, that is $r = (r_1, \dots, r_n)$, we define the $n \times n$ kernel matrix $\mathbf{K}_n(r)$:
\begin{equation}
    \mathbf{K}_n(r) =
    \begin{bmatrix}
        K(r_1, r_1) & \dots & K(r_1, r_n) \\
        \vdots & \ddots & \vdots \\
        K(r_n, r_1) & \dots & K(r_n, r_n)
    \end{bmatrix}    
\end{equation}
Using this kernel matrix, we specify the density of the Projection DPP as follows:
\begin{equation}
    \dpp(r; n) = \frac{1}{n!} \det \mathbf{K}_n(r)
    \label{eq:dpp}
\end{equation}
Since $\mathbf{K}_n(r)$ is positive semi-definite, it follows that its determinant is non-negative so that $\dpp(r; n)$ is non-negative, as desired.  A proof that $\dpp(r; n)$ is properly normalized (i.e.~integrates to $1$) can be found, for example, in Proposition 2.10 of \citep{johansson2006random}.

Given the notion of a Projection DPP, we may define the base density as follows.  As above, let the base random variable be $z$, where $z$ can be broken into spin-up and spin-down pieces, denoted $z_u$ and $z_d$.  (Specifically, $z_u$ and $z_d$ are the parts of $z$ corresponding to electrons in $\nup$ and $\ndn$, respectively.)  The base density can then be constructed by taking
\begin{equation}
    \rho_z(z) = \dpp(z_u; n_u) \dpp(z_d; n_d)
    \label{eq:base_density}
\end{equation}
That is, $z_u$ and $z_d$ are chosen from two independent Projection DPPs.  We then have the following theorem:
\begin{theorem}
    \label{thm:base_density}
    Let $\rho_z$ be the density specified in Equation (\ref{eq:base_density}).  Then $\rho_z$ satisfies conditions (R1) and (R2) from Theorem \ref{thm:spin_multiplicity}.
\end{theorem}
\proofmain{base_density}
We therefore have an explicit form for the base density from Equations (\ref{eq:dpp}) and (\ref{eq:base_density}).  Furthermore, sampling from the base density amounts to sampling from two independent Projection DPPs.  A sampling procedure for Projection DPPs is specified in \suppmat{}.

We note that it is possible to make the base density itself learnable.  This is achieved by making each Projection DPP learnable, through the introduction of learnable orthogonal functions $h_k : \mathbb{R}^D \to \mathbb{R}$, that is $h_k(y; \theta)$, from which the kernel function $K(y, y'; \theta) = H(y; \theta)^T H(y; \theta)$ becomes learnable.  In order to retain orthogonality, one may proceed as follows.  Let $B_\theta: \mathbb{R}^D \to \mathbb{R}^n$ be a specified by a network with parameters $\theta$, for example a Multilayer Perceptron (in which case $B_\theta$ is a type of neural field), and let $H(y; \theta) = AB_\theta(y)$ for a square $n \times n$ matrix $A$.  Let the Gram matrix of the network $B_\theta$ be given by $\Xi_\theta$, that is $(\Xi_\theta)_{ij} = \langle (B_\theta)_i, (B_\theta)_j \rangle$.  If the eigendecomposition of $\Xi_\theta$ is given by $\Xi_\theta = U_\theta \Lambda_\theta U_\theta^T$, then orthogonality of $H(\cdot; \theta)$ is achieved by choosing $A = A_\theta = \Lambda_\theta^{-1/2} U_\theta^T$.

\subsection[G-Equivariant Layers]{$\G$-Equivariant Layers}

As noted in Section \ref{sec:normalizing_flow}, we require the normalizing flow transformation to be $\G$-equivariant.  Of course, chaining together many layers which are each $\G$-equivariant results in an overall transformation which is also $\G$-equivariant.  Now, suppose that a particular layer $\l$ can be written as
\begin{equation}
    r^\lp = T^\l(r^\l)
\end{equation}
where $r^\l = (r_1^\l, \dots, r_n^\l)$ and likewise for $r^\lp$.  We will need to see the action on the spin-up and spin-down electrons separately, so we denote $r_u^\l = (r_i^\l)_{i \in \nup}$ and $r_d^\l = (r_i^\l)_{i \in \ndn}$; and we may write
\begin{equation}
    r_u^\lp = T_u^\l(r_u^\l, r_d^\l) \quad \text{and} \quad r_d^\lp = T_d^\l(r_u^\l, r_d^\l)
    \label{eq:layer_by_spin}
\end{equation}
For notational convenience, we use $\alpha \in \{u, d\}$ to denote the spin, and the complement of the spin is given by $\alphac$ (i.e.~if $\alpha = u$ then $\alphac=d$ and vice-versa).  Then we have the following theorem:
\begin{theorem}
    \label{thm:equivariance_by_spin}
    The transformation $T^\l$ is $\G$-equivariant if and only if 
    \begin{equation}
        \Tal(\pi_\alpha \ral, \pi_\alphac \rcl) = \pi_\alpha \Tal(\ral, \rcl) \qquad \alpha \in \{u, d\}
        \label{eq:layer_eq_inv}
    \end{equation}
    That is, $\Tal$ is equivariant with respect to $r_\alpha^\l$, and invariant with respect to $r_\alphac^\l$.
\end{theorem}
\proofmain{equivariance_by_spin}

We now show how to specify continuous and discrete normalizing flows satisfying Theorem \ref{thm:equivariance_by_spin}.

\subsection{Continuous Normalizing Flows}
\label{sec:continuous_flow}

According to Theorem \ref{thm:normalizing_flow}, we are required a find a transformation which is $\G$-equivariant.  We now show this can be achieved via a continuous normalizing flow.  We specify this flow via the ordinary differential equation (ODE)
\begin{equation}
    \frac{dv}{dt} = \Gamma_t(v), \quad \text{with} \quad  v(0) = z \sim \rho_z(\cdot) \quad \text{and} \quad r = v(1)
    \label{eq:continuous_flow}
\end{equation}
That is, the transformation $r = T(z)$ is derived as follows: the initial condition is sampled from the base density; and $r$ is gotten by integrating the ODE forward to time $t=1$.  $\Gamma$'s $t$-dependence is indicated via a subscript for notational convenience.  We then have the following theorem:
\begin{theorem}
    \label{thm:continuous_flow}
    Let the transformation $r = T(z)$ be specified as in Equation (\ref{eq:continuous_flow}).  Then $T$ is $\G$-equivariant if $\Gamma_t$ is $\G$-equivariant for all $t$.
\end{theorem}
\proofmain{continuous_flow}
It therefore suffices to design a $\G$-equivariant function $\Gamma_t$.  Let us break this down by spin: from Theorem \ref{thm:equivariance_by_spin}, we know that this implies that for all $t$, we have that $\Gamma_t(\pi_\alpha r_\alpha, \pi_\alphac r_\alphac) = \pi_\alpha \Gamma_t(r_\alpha, r_\alphac) \text{ for } \alpha \in \{u, d\}$.  We show in \suppmat{} how to implement a layer of $\Gamma$ with a combination of multihead attention, fully connected layers, and linear projections ($\Gamma$ can be composed of many such layers).

Continuous normalizing flows are elegant; however, they can present some numerical difficulties.  In particular, the issue of ODE stiffness frequently arises in deep learning pipelines involving continuous normalizing flows.  Thus, we now present an alternative method, based on discrete normalizing flows.

\subsection{Discrete Normalizing Flows}
\label{sec:discrete_normalizing_flows}

Our goal is now to design such functions $T_u^\l$ and $T_d^\l$ which satisfy Equation (\ref{eq:layer_eq_inv}), and for which the overall transformation $T^\l = (T_u^\l, T_d^\l)$ is invertible.  The goal of the layer we propose here is to \textit{not sacrifice on expressivity}, especially when compared to many layers which are designed for discrete normalizing flows.  In particular, the main issue will be to show that the expressivity can be retained even with the joint requirements of invertibility and $\G$-equivariance.  We note that the kind of transformation we propose below is not generally used for normalizing flows, as the determinant of its Jacobian is not fast to compute; however, this is not an issue in our case, as the dimension of the spaces we are dealing with is moderate in size.  For a more detailed discussion, see \suppmat{}.

To solve this problem, we introduce the Split Subspace Layer; we note that this layer may be of broader interest in machine learning, independent of the current setting.  As before, we take $D$ to represent the ambient spatial dimension; in our case, $D=3$.  A key parameter for the $\l^{th}$ layer will be the orthogonal matrix $\Lal \in O(D)$; in particular, we divide this matrix into 2 pieces
\begin{equation}
    \Lal = [\bal, \xal]
    \text{ with }
    \bal \in \mathbb{R}^{D \times D_\beta}
    \text{ and }
    \xal \in \mathbb{R}^{D \times (D-D_\beta)}
\end{equation}
That is, $\bal$ represents the first $D_\beta$ columns of $\Lal$, and $\xal$ represents the final $D-D_\beta$ columns.  For each electron $i$, we compute the inner product of its coordinates with $\bal$, i.e.
\begin{equation}
    \gail = (\bal)^T \rail \quad \text{so that } \gail \in \mathbb{R}^{D_\beta}
\end{equation}
We can collect the individual vectors $\gail$ into a list $\gal = ( \gail )_{i \in \na}$.
Given this, we define the Split Subspace Layer $T_\alpha^\l$ on a per-electron basis by
\begin{equation}
    \railp = \Tail(\ral, \rcl) = \rail + \xal \Pail(\gal, \gcl)
    \label{eq:layer_main}
\end{equation}
where $\Pal$ is a network, $\Pail$ is the part of (the output of) $\Pal$ corresponding to the $i^{th}$ electron, and $\Pail(\gal, \gcl) \in \mathbb{R}^{D-D_\beta}$.  A crucial aspect of this $\Pal$ network is that it captures dependencies between all electrons.  This can been seen by examining Equation (\ref{eq:layer_main}): $\Pail$ depends on both $\gal$ and $\gcl$.  But recall that $\gal$ is a list of the vectors $\gail$ for each electron $i$ with spin $\alpha$, and $\gcl$ is the corresponding list for electrons with the complement spin $\alphac$.   As a result, $\Pail$ depends on all electrons, as desired.

The layer is referred to as the Split Subspace Layer due to the fact that its input is one subspace of $\mathbb{R}^D$, given by $\bal$; whereas its output is in the orthogonal complement of this subspace, given by $\xal$.  Note that there are multiple possible versions of this layer, as any choice $D_\beta$ in the set $\{ 1, \dots, D-1 \}$ is valid.  Since in our case $D = 3$, this gives us exactly two choices: $D_\beta = 1$ or $D_\beta = 2$.

The main ingredient of the layer is the network $\Pal$.  We now show two things: (1) the layer is invertible for any choice of $\Pal$ (2) we derive conditions on $\Pal$ to achieve $\G$-equivariance of $\Tal$.
\begin{theorem}
    \label{thm:layer_inverse_equivariance}
    Let $T^\l$ be a Split Subspace Layer, as given in Equation (\ref{eq:layer_main}).  Then $T^\l$ is invertible.  In particular, let $\gailrp = (\bal)^T \railp$; then the inverse of the layer is given by
    \begin{equation}
        \rail = \railp - \xal \Pail(\galrp, \gclrp)
    \end{equation}
    Furthermore, the layer $T^\l$ is $\G$-equivariant if
    \begin{equation}
        \Pal(\pi_\alpha \gal, \pi_\alphac \gcl) = \pi_\alpha \Pal(\gal, \gcl)
        \label{eq:phi_equivariance}
    \end{equation}
    i.e.~if $\Pal(\gal, \gcl)$ is equivariant with respect to permutations on $\gal$ and invariant with respect to permutations on $\gcl$.
\end{theorem}
\proofmain{layer_inverse_equivariance}
The Split Subspace Layer therefore depends on implementation of the network $\Pal$ so that it satisfies Equation (\ref{eq:phi_equivariance}).  We show in \suppmat{} how $\Pal$ can be implemented with a combination of multihead attention, fully connected layers, and linear projections.  We specify a more general version of the Split Subspace Layer in \suppmat{}.

We now comment on the complexity of this method vs. that of standard MCMC approaches to sampling.  In our case, a single sample may be drawn by passing a sample from the base density through the flow network; in other words, only a single call to the flow network is required.  By contrast, when sampling using MCMC techniques, each sample requires many calls to the network representing the wavefunction or density.  For example, in using Langevin Monte Carlo techniques, a network representing the (unnormalized) density $\rho(r)$ may be sampled by the iterations
\begin{equation}
    r_{k+1} = r_k + \tau \nabla \log \rho(r_k) + \sqrt{2\tau} \xi_k
\end{equation}
where $\xi_k$ is Gaussian noise.  In the limit of $\tau \to 0$, as $k \to \infty$, the samples thus generated will be distributed according to (the normalized version of) the density $\rho(r)$.  We note that in the case of finite $\tau$, one often adds treats each new iterate as a Metropolis-Hastings style proposal, which can be correspondingly accepted or rejected.  The key point, however, is that a single sample requires many calls to the network representing $\rho$, whereas in our case only a single network call is required.

\subsection{Training via SGD}
\label{sec:training_via_sgd}

\parbasic{Log Domain: Density}
In order to avoid numerical issues, it is best to operate in the log domain.  Suppose that 
\begin{equation}
    \psi(r) = e^{q(r) + iw(r)}
\end{equation}
so that
\begin{equation}
    q(r) = \tfrac{1}{2} \log \rho(r)
    \hspace{2.0mm} \text{and} \hspace{2.0mm}
    w(r) = \text{atan2}\left( \kappa_i(r), \kappa_r(r) \right)
\end{equation}
where $\kappa_r(r)$ and $\kappa_i(r)$ are the real and imaginary parts of the phase $\kappa(r)$, respectively; and atan2 is the ``full'' arctangent.

The log-density $q(r; \theta)$ may be computed for both continuous and discrete normalizing flows, where we now introduce the parameters $\theta$ of the network explicitly.  Consider a sample $z$ chosen from the base density $\rho_z(z)$, and in analogy to $q(r)$, define $q_z(z) = \tfrac{1}{2} \log \rho_z(z)$.  Now, in the case of a continuous normalizing flow, let $v(t)$ satisfy Equation (\ref{eq:continuous_flow}); then $q(r; \theta)$ can be by computed \citep{chen2018neural} by solving the ODE
\begin{equation}
    \frac{da}{dt} = -\text{Trace} \left( \frac{\partial \Gamma_t}{\partial v}(v(t); \theta) \right)
\end{equation}
with $a(0) = q_z(z)$ and $q(r; \theta) = a(1)$.  This is the continuous analogue of the change of variables formula.
In the case of a discrete normalizing flow, fix the following notation: $r^0 = z$, $r = r^{L+1}$, and $T = T^L \circ \dots \circ T^0$.  Then we may use a logarithmic version of the standard change of variables formula (\ref{eq:change_of_variables}):
\begin{equation}
    q(r; \theta) = q_z \left( T^{-1} (r; \theta) \right) + \frac{1}{2} \sum_{\l=0}^L \log \left| \det J_{(T^\l)^{-1}}(r^\lp; \theta) \right| \label{eq:q}
\end{equation}

\parbasic{Log Domain: Gradient of the Objective}
Recall that our goal in finding an approximation to the ground state wavefunction is to solve the optimization problem in Equation (\ref{eq:loss}).  Using Equation (\ref{eq:loss_for_sgd}) and noting that $\langle \psi(\cdot; \theta) | \psi(\cdot; \theta) \rangle = 1$ since $\rho(\cdot; \theta)$ is normalized, we may write the objective function to be minimized as
\begin{align}
    \mathcal{L}(\theta)
    & = \langle \psi(\cdot; \theta) | H | \psi(\cdot; \theta) \rangle \notag \\
    & = \mathbb{E}_{r \sim \rho(\cdot; \theta)} \left[ \elocal_r(r; \theta) \right]
    \, \approx \, \frac{1}{K} \sum_{k=1}^K \elocal_r\left(r^{(k)}; \theta \right)
    \label{eq:loss_redux}
\end{align}
with samples $r^{(k)} \sim \rho(\cdot; \theta)$.  Then we have the following theorem, which shows that the local energy can be written entirely as a function of $q(r; \theta)$ and the potential $V(r)$, so that the phase $w(r; \theta)$ does not appear; and furthermore gives the gradient of the objective function $\mathcal{L}(\theta)$.
\begin{theorem}
    \label{thm:log_domain}
    The local energy can be written as
    \begin{equation}
        \elocal_r(r; \theta) = -\tfrac{1}{2} \Delta_r q(r; \theta) -\tfrac{1}{2} \| \nabla_r q(r; \theta) \|^2 + V(r)
        \label{eq:elocal}
    \end{equation}
    In particular, the local energy is independent of the phase $w(r; \theta)$.  Furthermore, let
    \begin{equation}
        \Omega(r; \theta) = 2 \left( \elocal_r(r; \theta) - \mathbb{E}_{r \sim \rho(\cdot; \theta)} \left[ \elocal_r(r; \theta) \right] \right) \nabla_\theta q(r; \theta)
        \label{eq:omega}
    \end{equation}
    Then the gradient of the loss function may be written as
    \begin{equation}
        \nabla_\theta \mathcal{L}(\theta) = \mathbb{E}_{r \sim \rho(\cdot; \theta)} \left[ \Omega(r; \theta) \right]
                                          \approx \frac{1}{K} \sum_{k=1}^K \Omega \left(r^{(k)}; \theta \right)
        \label{eq:loss_gradient}
    \end{equation}
    with samples $r^{(k)} \sim \rho(\cdot; \theta)$.
\end{theorem}
\proofmain{log_domain}
Thus, in order to optimize the objective in Equation (\ref{eq:loss_redux}), we may use gradient descent using the estimate for the gradient in Equation (\ref{eq:loss_gradient}).  A detailed version of the optimization routine is given in \suppmat{}.

\section{Further Details: Phase, Cusps, and Induction}
\label{sec:further_details}

\subsection{The Phase}
\label{sec:phase}

Since the Hamiltonian is time-reversal invariant  and  Hermitian, both its eigenvalues and its eigenfunctions are real.  Since the ground-state wavefunction we are looking for is real, the phase can be taken to belong to the two element set $\{0, \pi\}$.  Given that we now know how to solve for an approximation to the density $\rho_0(r)$ corresponding to the ground state wavefunction, we now show one way of assigning the phase so that the resulting ground state wavefunction $\psi_0(r)$ is appropriately antisymmetric.
\begin{theorem}
    \label{thm:phase}
    Let $\rho_0(r)$ be the the density for the ground state wavefunction. Let $\prec$ be a strict total order on $\mathbb{R}^D$, and define the set
    \begin{align}
        \mathcal{R} =
        \{r = (r_1, \dots r_n):
        & \hspace{1.0mm} r_1 \prec r_2 \prec \dots \prec r_{n_u} 
        \hspace{2.0mm} \text{and} \notag \\
        & r_{n_u+1} \prec r_{n_u+2} \prec \dots \prec r_n \}
    \end{align}
    For any $r$ without $r_i = r_j$, define the permutation $\pir \in \G$ by $\pir r \in \mathcal{R}$.  Then a valid antisymmetric ground state wavefunction is given by
    \begin{equation}
        \psi_0(r) = 
        \begin{cases}
            (-1)^\pir \sqrt{\rho_0(r)} & \text{if } r_i \neq r_j \,\, \forall i, j \\
            0 & \text{otherwise}
        \end{cases}
        \label{eq:phase_integrated}
    \end{equation}
\end{theorem}
\proofmain{phase}
Thus, given the density $\rho_0$, we can use Theorem \ref{thm:phase} to easily compute the ground state wavefunction $\psi_0$.  A question remains: what is the strict total order $\prec$?  Any choice is valid, but the simplest thing to do is to use lexicographic ordering on the coordinates of the two points in $\mathbb{R}^D$ that are being compared.


\subsection{Incorporating Cusps}

\parbasic{Electron-Electron Cusps}
Wavefunctions are known to have certain non-smooth properties, known as cusps.  In particular, the gradient of the wavefunction should exhibit a discontinuity when two electrons coincide.  One way to incorporate such gradient discontinuities is via the introduction of terms which depend on the distance between electrons \citep{pfau2020ab}; as the distance is itself a continuous but non-smooth function of the electron positions, using distances can allow us to model such cusps.
In the case of the discrete normalizing flow, our goal will be to design a layer which incorporates the inter-electron distances directly.  Given the requirements of a normalizing flow, the challenge is to enforce invertibility for such a layer.  We have the following result:
\begin{theorem}
    \label{thm:cusps}
    Let the set of distances be given by $\Dl = \left\{ \Dijl \right\}_{i < j}$ where $\Dijl = \| \ril - \rjl \|$.  Given a layer of the form
    \begin{equation}
        \rilp = \Rl(\Dl; \theta) \, \ril + \tl(\Dl; \theta)
        \label{eq:cusps}
    \end{equation}
    with $\Rl(\Dl; \theta) \in O(D)$ and $\tl(\Dl; \theta) \in \mathbb{R}^D$. Then the layer is both $\G$-equivariant as well as invertible.
\end{theorem}
\proofmain{cusps}
The essence of this layer to rotate all electrons in a given configuration $r = (r_1, \dots, r_n)$ by the same rotation matrix $\Rl$ and translation vector $\tl$; and the rotation matrix and translation vector are both functions the configuration $r$ entirely through the distances $\Dl$.  The latter fact is crucial, as it means that different configurations $r$ are treated differently, which gives the layer expressivity.  An implementation of this layer based on a Deep Set architecture \citep{zaheer2017deep} is given in \suppmat{}.

It is also known that the gradient of the wavefunction should exhibit a discontinuity when an electron and nucleus coincide.  The treatment is similar, and is given in \suppmat{}.

\subsection{Induction Across Multiple Molecules}
\label{sec:multiple_molecules}

In an effort to accelerate the ground state computation, we may try to learn the ground state wavefunctions and energies for an entire class of molecules simultaneously, as in \citep{gao2023generalizing, scherbela2024towards, gerard2024transferable}.  In particular, the molecular parameters are given by $R = (R_1, \dots, R_N)$, the nuclear positions; and $Z = (Z_1, \dots, Z_N)$, the atomic numbers of each nucleus.  Then our goal is to learn a function of the form $\psi_0(x; R, Z)$, i.e.~a ground state wavefunction which is explicitly parameterized by the molecular parameters.  This entails computing the density $\rho(r; R, Z)$.  However, this latter task is made more complicated by the fact that two new invariances are required.  The first is nuclear permutation invariance:
\begin{equation}
    \rho(r; \pi R, \pi Z) = \rho(r; R, Z) \quad \text{for } \pi \in \mathbb{S}_N
    \label{eq:nuclear_permutation_invariance}
\end{equation}
and the second is joint rigid motion invariance:
\begin{equation}
    \rho(\tau r; \tau R, Z) = \rho(r; R, Z) \quad \text{for } \tau \in E(D)
    \label{eq:joint_rigid_invariance}
\end{equation}
We henceforth assume that the nuclei have their center of mass at the origin, i.e. $\bar{R} = \frac{1}{N}\sum_{I=1}^N R_I = 0$; this removes the need to deal with translations, which generally require special (and uninteresting) treatment, e.g. see \citep{satorras2021n}.  Thus, Equation (\ref{eq:joint_rigid_invariance}) becomes joint rotation invariance:
\begin{equation}
    \rho(\Theta r; \Theta R, Z) = \rho(r; R, Z) \quad \text{for } \Theta \in O(D)
    \label{eq:joint_rotation_invariance}
\end{equation}
We now show that densities satisfying Equations (\ref{eq:nuclear_permutation_invariance}) and (\ref{eq:joint_rotation_invariance}) can be realized via a variation of the continuous normalizing flow we have introduced in Section \ref{sec:continuous_flow}:
\begin{theorem}
    \label{thm:induction}
    Let $\bar{R} = \frac{1}{N}\sum_{I=1}^N R_I = 0$.  Given a continuous normalizing flow of the form $dv/dt = \Gamma_t(v; R, Z)$ with $v(0) = z \sim \rho_z(\cdot)$ and $r = v(1)$.  Let the function $\Gamma_t$ be invariant with respect to nuclear permutations and equivariant with respect to joint rotations, i.e. for all $t$
    \begin{align}
        & \Gamma_t(v; \pi R, \pi Z) = \Gamma_t(v; R, Z) \hspace{2.0mm} \forall \pi \in \mathbb{S}_N \notag \\
        & \Gamma_t(\Theta v; \Theta R, Z) = \Theta \Gamma_t(v; R, Z) \hspace{2.0mm} \forall \Theta \in O(D)
        \label{eq:gamma_extra_properties}
    \end{align}
    Furthermore, suppose that the base density is invariant with respect to rotations, $\rho_z(\Theta z) = \rho_z(z)$ for $\Theta \in O(D)$.  Then the resulting density $\rho(r; R, Z)$ satisfies Equations (\ref{eq:nuclear_permutation_invariance}) and (\ref{eq:joint_rotation_invariance}).
\end{theorem}
\proofmain{induction}
First, we note that the base density in Equation (\ref{eq:base_density}) can be made invariant to rotations by constructing the relevant Projection DPP from a kernel function $K(y, y') = H(y)^T H(y)$, where the functions $h_i(y)$ are derived from taking arbitrary rotationally-invariant functions $\tilde{h}_i(y)$, and orthogonalizing them with Gram-Schmidt; e.g.~one may use Gaussians of varying bandwidths, $\tilde{h}_i(y) = e^{-\|y\|^2 / \sigma_i^2}$.

Now, we turn to the construction of $\Gamma_t$.  Recall from Theorem \ref{thm:continuous_flow} that $\Gamma_t(\cdot; R, Z)$ must be $\G$-equivariant for all $t$.  Furthermore, we have already noted that $\G$-equivariant functions may be constructed using a combination of standard pieces: multihead attention, fully connected layers, and linear projections.  It would be nice if we were able to use this result while also incorporating the extra conditions in Equation (\ref{eq:joint_rigid_invariance}).  We now show that this is possible:
\begin{theorem}
    \label{thm:induction_transformation}
    Let $\phi_t(v; R, Z)$ be a function which is $\G$-equivariant with respect to $v$ i.e.~$\phi_t(g v; R, Z) = g \phi_t(v; R, Z)$ for $g \in \G$.  Let $\omega_t(v; R, z)$ be a function whose output is itself a rotation, i.e. $\omega_t(v; R, z) \in O(D)$.  Let $\omega_t$ be $\G$-invariant with respect to $v$, and $O(D)$-equivariant jointly with respect to $v$ and $R$ i.e.~$\omega_t(\Theta v; \Theta R, Z) = \Theta \omega_t(v; R, Z)$.  Finally, let both $\phi_t$ and $\omega_t$ be permutation-invariant jointly with respect to $R$ and $Z$ i.e.~$\phi_t(v; \pi R, \pi Z) = \phi_t(v; R, Z)$ and likewise for $\omega_t$. 
    Then the function 
    \begin{equation}
        \Gamma_t(v; R, Z) = \zeta \phi_t(\zeta^{-1} v; \zeta^{-1} R, Z)
    \end{equation}
    where $\zeta = \omega_t(v; R, Z)$ satisfies the properties in Equation (\ref{eq:gamma_extra_properties}) and is $\G$-equivariant with respect to $v$.
\end{theorem}
\proofmain{induction_transformation}
We can use the previously mentioned recipe in \suppmat{} in order to construct a $\G$-equivariant $\phi_t$, with an extra path in the network for the $R, Z$ dependence, based on either Deep Set or a Transformer architecture with pooling to gain the requisite invariance.  The function $\omega_t$ can be constructed by using an $E(D)$ Equivariant Graph Neural Network \citep{satorras2021n} whose output is a rotation matrix, similar to what is done in \citep{kaba2023equivariance}.  More detailed information is contained in \suppmat{}.

\section{Conclusions}
\label{sec:conclusions}
We have demonstrated a theoretical framework for efficiently solving the Electronic \Schrodinger{} Equation using normalizing flows.  Using these flows allows us to sample efficiently from the wavefunction, thereby side-stepping the need for time-consuming MCMC approaches to sampling.  Future work will focus on using diffusion techniques \citep{yang2023diffusion} to model wavefunctions, which is not straightforward due to the need to maintain the requisite quantum mechanical properties; and on adapting flow-matching \citep{lipman2022flow} and related techniques for the optimization of the variational objective.

\newpage
{\LARGE \sc \textbf{Appendix}}
\appendix

\section{Derivation of Equation (\ref{eq:loss_for_sgd})}
\label{app:derivation_of_loss_for_sgd}

Recall that the local energy is defined as
\begin{equation}
    \elocal(x) \equiv \frac{H \psi(x)}{\psi(x)} = -\frac{\Delta \psi(x)}{2\psi(x)} + V(x) 
\end{equation}
with
\begin{equation}
    \elocal_r(x) = \text{Real}\{ \elocal(x) \}
\end{equation}
In this case, one may write
\begin{align}
    \frac{\langle \psi | H | \psi \rangle}{\langle \psi | \psi \rangle}
    & = \text{Real}\left\{ \frac{\langle \psi | H | \psi \rangle}{\langle \psi | \psi \rangle} \right\} \notag \\
    & = \text{Real}\left\{ \frac{\int \psi^*(x) H \psi(x) dx}{\langle \psi | \psi \rangle} \right\} \notag \\
    & = \frac{1}{\langle \psi | \psi \rangle} \int \text{Real}\left\{ \psi^*(x) H \psi(x) \right\} dx \notag \\
    & = \frac{1}{\langle \psi | \psi \rangle} \int \text{Real}\left\{ \frac{\psi(x)}{\psi(x)} \psi^*(x) H \psi(x)  \right\} dx \notag \\
    & = \frac{1}{\langle \psi | \psi \rangle} \int \text{Real}\left\{ |\psi(x)|^2 \frac{H \psi(x)}{\psi(x)}  \right\} dx \notag \\
    & = \int \text{Real}\left\{ \frac{H \psi(x)}{\psi(x)}  \right\} \frac{|\psi(x)|^2}{\langle \psi | \psi \rangle} dx \notag \\
    & = \int  \elocal_r(x) \rho(x) dx \notag \\
    & = \mathbb{E}_{x \sim \rho(\cdot)} \left[ \elocal_r(x) \right] \notag \\
    & \approx \frac{1}{K} \sum_{k=1}^K \elocal_r\left(x^{(k)}\right)
\end{align}
where the $x^{(k)}$ are sampled from $\rho(\cdot) = |\psi(\cdot)|^2 / \langle \psi | \psi \rangle$.  Note that in the first line, we have used the fact that $H$ is a symmetric operator so that the quadratic form is real; in the third line, the fact that $\langle \psi | \psi \rangle$  is real; and in the sixth line, the fact that $|\psi(x)|^2$ is real.

\section{Proof of Theorem \ref{thm:density_phase}}
\label{prf:density_phase}

\textbf{Theorem.}
\textit{
    Let $\rho(\cdot)$ be a probability density function which we can sample from in constant time.  Let $\rho(\cdot)$ satisfy two additional properties:
    \begin{enumerate}[label=(D\arabic*)]
        \item $\rho(x)$ is symmetric: $\rho(\pi x) = \rho(x)$ for all permutations $\pi \in \mathbb{S}_n$.
        \item $\rho(x) = 0$ if $x_i = x_j$ for any $i, j$.
    \end{enumerate}
    Finally, let $\kappa(x)$ be a complex function which satisfies $|\kappa(x)| = 1 \,\, \forall x$, and is nearly antisymmetric: 
    \begin{equation*}
        \kappa(\pi x) = 
        \begin{cases}
            (-1)^\pi \kappa(x) & \text{if } x_i \neq x_j \text{ for all } i, j \\
            \bar{\kappa} & \text{otherwise}
        \end{cases}
    \end{equation*}
    where $\bar{\kappa} \in \mathbb{C}$ is an arbitrary value with $|\bar{\kappa}|=1$.
    Then $\psi$ satisfies (W1)-(W4) if and only if $\psi$ can be written as $\psi(x) = \kappa(x) \sqrt{\rho(x)}$ with $\kappa$ and $\rho$ satisfying the above-stated properties.
}
\begin{proof}
    Suppose that $\psi(x) = \kappa(x) \sqrt{\rho(x)}$, let us prove each of properties (W1)-(W4).

    (W1) The functional form for $\psi(x)$ is just $\kappa(x) \sqrt{\rho(x)}$, which we know explicitly.

    (W2) Antisymmetry of $\psi$: we break down by cases.  Suppose that $x$ is such that $x_i \neq x_j \text{ for all } i, j$.  Then:
    \begin{align}
        \psi(\pi x) & = \kappa(\pi x) \sqrt{\rho(\pi x)} \notag \\
                    & = (-1)^\pi \kappa(x) \sqrt{\rho(x)} \notag \\
                    & = (-1)^\pi \psi(x)
    \end{align}
    where in the second line we have used the two facts that $\kappa$ is antisymmetric and $\rho$ is symmetric.  Now, suppose that $x_i = x_j$ for some $i, j$:
    \begin{align}
        \psi(\pi x) & = \kappa(\pi x) \sqrt{\rho(\pi x)} \notag \\
                    & = \bar{\kappa} \sqrt{\rho(x)} \notag \\
                    & = 0
    \end{align}
    where in the third line, we have used (D2).  But this is precisely what is required for an antisymmetric function: if $a(x)$ is antisymmetric and $x_i = x_j$ of some $i, j$, then $\pi_{ij} x = x$, where $\pi_{ij}$ is the permutation which flips $i$ and $j$, so that
    \begin{equation}
        a(\pi_{ij} x) = a(x)
        \, \Rightarrow \, -a(x) = a(x)
        \, \Rightarrow \, a(x) = 0
        \, \Rightarrow \, a(\pi x) = (-1)^\pi a(x) = 0
    \end{equation}
    where we have used the fact that $(-1)^{\pi_{ij}} = -1$, since only one flip is required.

    (W3) We can sample in constant time from $\rho(x) = \|\psi(x)\|^2$ by assumption.

    (W4) $\psi$ is normalized:
    \begin{equation}
        \int \| \psi(x) \|^2 dx = \int |\kappa(x)|^2 \rho(x) dx = \int \rho(x) dx = 1
    \end{equation}
    since $\rho(x)$ is a probability density function.
    
    Thus, we have proved the forward direction.

    Now, let us assume properties (W1)-(W4).  We can always express a complex number $c$ in terms of a magnitude and a phase; in particular, we may write $c = me^{i\nu}$, where $m \ge 0$ is a real number, and $\nu \in [0, 2\pi)$.  (W1) tells us that we have an explicit form for the complex-valued function $\psi(x)$; thus, we know that
    \begin{equation}
        \psi(x) = m(x) e^{i \nu(x)} \equiv \sqrt{\rho(x)} \kappa(x)
    \end{equation}
    with $|\kappa(x)| = 1$, where we have used the fact that $m(x) \ge 0$.  Note that $\rho(x) = \|\psi(x)\|^2$.  As $\rho(x) \ge 0$, and 
    \begin{equation}
        \int \rho(x) dx = \int \|\psi(x)\|^2 dx = 1
    \end{equation}
    by (W4), then $\rho(x)$ is a density.  Furthermore, by (W3) $\rho(x) = \|\psi(x)\|^2$ may be sampled in constant time.  Finally, by (W2), $\psi(x)$ is antisymmetric; thus,
    \begin{align}
        & \psi(\pi x) = (-1)^\pi \psi(x) \quad \forall \pi, x \notag \\
        \Leftrightarrow \quad & \sqrt{\rho(\pi x)} \kappa(\pi x) = (-1)^\pi \sqrt{\rho(x)} \kappa(x) \quad \forall \pi, x
    \end{align}
    In cases where $\pi x = x$ and $(-1)^{\pi} = -1$, then we have that
    \begin{align}
        & \sqrt{\rho(\pi x)} \kappa(\pi x) = (-1)^\pi \sqrt{\rho(x)} \kappa(x) \notag \\
        \Leftrightarrow \quad & \sqrt{\rho(x)} \kappa(x) = - \sqrt{\rho(x)} \kappa(x) \notag \\
        \Leftrightarrow \quad & \rho(x) = 0
    \end{align}
    where the third line follows from the fact that $\kappa(x) \neq 0$ since $|\kappa(x)|$ = 1.  However, note that $\pi x = x$ and $(-1)^{\pi} = -1$ is equivalent to $x_i = x_j$ for some $i, j$.  Thus, we have established (D2).  Furthermore, in such cases we can take $\kappa(x) = \bar{\kappa}$, since it plays no role.  In all other cases, i.e.~where $x_i \neq x_j \text{ for all } i, j$, we have that
    \begin{align}
        & \sqrt{\rho(\pi x)} \kappa(\pi x) = (-1)^\pi \sqrt{\rho(x)} \kappa(x) \quad \forall \pi, x \text{ such that } x_i \neq x_j \text{ for all } i, j \notag \\
        \Leftrightarrow \quad & \sqrt{\rho(\pi x)} = \sqrt{\rho(x)} \quad \text{and} \quad  \kappa(\pi x) = (-1)^\pi \kappa(x) \quad \forall \pi, x \text{ such that } x_i \neq x_j \text{ for all } i, j
    \end{align}
    where the second line holds since this must hold true for all $\pi$ and all relevant $x$.  $\kappa(\pi x) = (-1)^\pi \kappa(x)$ establishes the remainder of the nearly antisymmetric character of $\kappa$.  Finally, 
    \begin{equation}
        \sqrt{\rho(\pi x)} = \sqrt{\rho(x)} \quad \Leftrightarrow \quad \rho(\pi x) = \rho(x) \quad \forall \pi, x \text{ such that } x_i \neq x_j \text{ for all } i, j
    \end{equation}
    since $\rho(x) \ge 0$.  This shows that $\rho(x)$ is symmetric for the case $x_i \neq x_j \text{ for all } i, j$. Indeed $\rho(x)$ is symmetric for all $x$, including those for which $x_i = x_j$ for some $i, j$, as in the latter case we have shown that $\rho(x) = 0$.  This establishes (D1) and completes the proof.
\end{proof}

\section{Proof of Theorem \ref{thm:spin_multiplicity}}
\label{prf:spin_multiplicity}

\begin{theorem*}
    Given a configuration $x = (r, s)$, let a permutation which maps the spin vector $s$ to the canonical spin vector $\bar{s}$ be given by $\bar{\pi}_s$, i.e. $\bar{s} = \bar{\pi}_s s$.  Let $\bar{\rho}(r)$ be a density function on electron positions (i.e.~no spins) satisfying
    \begin{enumerate}[label=(R\arabic*)]
        \item $\bar{\rho}$ is $\G$-invariant: $\bar{\rho}(\pi r) = \bar{\rho}(r) \text{ for all } \pi \in \G$
        \item $\bar{\rho}(r) = 0 \text{ if } r_i = r_j, \text{ for } i, j \in \nup \text{ or } i, j \in \ndn$
    \end{enumerate}
    A density $\rho(x) = \rho(r, s)$ satisfies conditions (D1)-(D2) in Theorem 1 
    if and only if it may be written as  $\rho(r, s) = \bar{\rho}(\bar{\pi}_s r)$ for a density  $\bar{\rho}(r)$ satisfying conditions (R1) and (R2).
\end{theorem*}

\begin{proof}
Let us prove the forward direction: assume a density $\rho(x) = \rho(r, s)$ satisfying conditions (D1) and (D2), and we will show that it must be written as $\rho(r, s) = \bar{\rho}(\bar{\pi}_s^{-1} r)$ for $\bar{\rho}$ satisfying conditions (R1) and (R2).  Let $\bar{r} = \bar{\pi}_s r$ and $\bar{x} = \bar{\pi}_s x = (\bar{r}, \bar{s})$.  In this case, we have that
\begin{equation}
    \rho(x) = \rho(\bar{\pi}_s x) = \rho(\bar{x})
    \label{eq:rho_bar_x}
\end{equation}
where the first equality comes our requirement that $\rho(x)$ satisfy condition (D1), and the second equality from the definition of $\bar{x}$.  As a result, it is sufficient for us to focus on constructing a density $\rho(\bar{x}) = \rho(\bar{r}, \bar{s})$, i.e. a density where the spins are in canonical order.  As $\bar{s}$ is fixed as the canonical ordering, we may suppress it, writing $\rho(\bar{x}) = \rho(\bar{r}, \bar{s}) \equiv \bar{\rho}(\bar{r})$ for a function $\bar{\rho}$.  Now, $\bar{\rho}$ must satisfy condition (D1); however, the only permutations that are relevant are those that preserve the canonical spin ordering $\bar{s}$.  More specifically, the relevant permutations $\pi$ are those for which $\pi \bar{s} = \bar{s}$; it is easy to see that those permutations form the group $\mathbb{S}_{\nup} \times \mathbb{S}_{\ndn} = \G$.  Thus, we must have that
\begin{equation}
    \bar{\rho}(\pi \bar{r}) = \bar{\rho}(\bar{r}) \text{ for all } \pi \in \G
\end{equation}
That is, $\bar{\rho}$ is $\G$-invariant, which is condition (R1).

Now let us turn to condition (D2), which states that $\rho(x) = 0$ if $x_i = x_j$ for any $i, j$.  The requirement $x_i = x_j$ implies that both $r_i = r_j$ and $s_i = s_j$; and the condition $s_i = s_j$ is equivalent to $i, j \in \nup$ or $i, j \in \ndn$.  Thus, condition (D2) is equivalent to
\begin{equation}
    \bar{\rho}(\bar{r}) = 0 \text{ if } \bar{r}_i = \bar{r}_j, \text{ for } i, j \in \nup \text{ or } i, j \in \ndn
\end{equation}
which is simply condition (R2).

Thus, we have proven conditions (R1) and (R2) must hold.  Finally, using Equation (\ref{eq:rho_bar_x}) and the definitions of $\bar{\rho}$ and $\bar{r}$, we have that
\begin{equation}
    \rho(x) = \rho(\bar{x}) = \bar{\rho}(\bar{r}) = \bar{\rho}(\bar{\pi}_s r)
\end{equation}
which completes the proof of the forward direction.

Now, let us prove the reverse direction: assume a density $\bar{\rho}$ satisfying conditions (R1) and (R2), and we will show that $\rho(r, s) = \bar{\rho}(\bar{\pi}_s r)$ satisfies conditions (D1) and (D2).  Let us begin by computing $\bar{\pi}_{\pi s}$, which will prove useful in what follows.  $\bar{\pi}_{\pi s}$ is defined by $\bar{s} = \bar{\pi}_{\pi s} \pi s$.  However, we also know that $\bar{s} = \bar{\pi}_s s$; setting these equal gives
\begin{equation}
    \bar{\pi}_{\pi s} \pi s = \bar{\pi}_s s
    \quad \Rightarrow \quad
    \bar{\pi}_{\pi s} \pi = \bar{\pi}_s \hat{\pi}
\end{equation}
where $\hat{\pi}$ is some permutation leaves $s$ unchanged, i.e.~such that $\hat{\pi} s = s$.  Thus
\begin{equation}
    \bar{\pi}_{\pi s} = \bar{\pi}_s \hat{\pi} \pi^{-1}
\end{equation}
However, we know that $\bar{s} = \bar{\pi}_s s$ so that $s = \bar{\pi}_s^{-1} \bar{s}$.  Using the fact that $\hat{\pi} s = s$ gives
\begin{equation}
    \hat{\pi} \bar{\pi}_s^{-1} \bar{s} = \bar{\pi}_s^{-1} \bar{s}
    \quad \Rightarrow \quad
    \hat{\pi} \bar{\pi}_s^{-1} = \bar{\pi}_s^{-1} \breve{\pi}
\end{equation}
where $\breve{\pi}$ is some permutation that leaves $\bar{s}$ unchanged; which precisely implies that $\breve{\pi} \in \G$.  Rearranging gives
\begin{equation}
    \hat{\pi} = \bar{\pi}_s^{-1} \breve{\pi} \bar{\pi}_s
    \quad \Rightarrow \quad
    \bar{\pi}_{\pi s} = \bar{\pi}_s \bar{\pi}_s^{-1} \breve{\pi} \bar{\pi}_s \pi^{-1}
                      = \breve{\pi} \bar{\pi}_s \pi^{-1}
\end{equation}
Now, for any permutation $\pi$, we have that
\begin{equation}
    \rho(\pi x) = \bar{\rho}(\bar{\pi}_{\pi s} \pi r)
                = \bar{\rho}(\breve{\pi} \bar{\pi}_s \pi^{-1} \pi r)
                = \bar{\rho}(\breve{\pi} (\bar{\pi}_s r))
                = \bar{\rho}(\bar{\pi}_s r)
                = \rho(x)
\end{equation}
where in the second last equality, we have used the fact that $\breve{\pi} \in \G$, and that $\bar{\rho}$ is $\G$-invariant by (R1).  Thus, we have established property (D1), i.e.~that $\rho$ is symmetric.
%
%

Now, turning to condition (D2), let us consider an $x$ such that $x_i = x_j$ for a particular $i, j$.  Continuing to use the notation $\bar{r} = \pi_s r$, this implies that $\bar{r}_i = \bar{r}_j$ for either $i, j \in \nup$ or $i, j \in \ndn$.  Thus, by condition (R2), we have that $\bar{\rho}(\bar{r}) = 0$.  But then 
\begin{equation}
    \rho(x) = \bar{\rho}(\bar{\pi}_s r) = \bar{\rho}(\bar{r}) = 0
\end{equation}
which is precisely condition (D2); this completes the proof.
\end{proof}

\section{Proof of Theorem \ref{thm:normalizing_flow}}
\label{prf:normalizing_flow}

\begin{theorem*}
    Suppose that we have a normalizing flow, whose base density $\rho_z$ satisfies properties (R1) and (R2) from Theorem 2, 
    and whose transformation $T$ is $\G$-equivariant.  Then the density resulting from the normalizing flow will satisfy properties (R1) and (R2).
\end{theorem*}

\begin{proof}
    Let us begin by proving that the density resulting from the normalizing flow will satisfy condition (R1).  Theorem 1 in \citep{kohler2020equivariant} states the following: ``Let $\rho$ be a density on $\mathbb{R}^m$ which is $\G$-invariant and $\G > \mathbb{H}$.  If $f$ is an $\mathbb{H}$-equivariant diffeomorphism, then $\rho_f$, the push-forward of $\rho$ along $f$, is $\mathbb{H}$-invariant.''  In our instance, we may take $\mathbb{H} = \G$, and thereby have established that the density resulting from the normalizing flow is $\G$-invariant, thus satisfying condition (R1).
    
    We now turn to proving that the density resulting from the normalizing flow will satisfy condition (R2).  Suppose that the random variable for the base density is given by $z$, with density $\rho_z$; and the normalizing flow is given by transformation $T$, i.e.~$r = T(z)$.  Then by the change of variables formula, we know that the density of $r$ is given by $\rho_r(r) = \rho_z(T^{-1}(r))|\det J_{T^{-1}}(r)|$.  Now, we are interested in the case when $r_i = r_j$ for $i, j \in \nup$ (we may equally consider the case of $\ndn$, they are identical).  Let $\pi_{ij} \in \G$ be the permutation whose only action is to flip the coordinates of electrons $i$ and $j$.  Given that $r_i = r_j$, then by definition we have that $\pi_{ij} r = r$.  In this case, we have that
    \begin{equation}
        z \equiv T^{-1}(r) = T^{-1}(\pi_{ij} r) = \pi_{ij} T^{-1}(r)
        \label{eq:proof1}
    \end{equation}
    where the latter equality is due to the $\G$-equivariance of $T^{-1}$, which follows straightforwardly from the $\G$-equivariance of $T$.  Rearranging the above, we have that
    \begin{equation}
        T^{-1}(r) = \pi_{ij}^{-1} z = \pi_{ij} z
        \label{eq:proof2}
    \end{equation}
    where the second equality is due to the fact that $\pi_{ij}^{-1} = \pi_{ij}$, as $\pi_{ij}$ simply flips electrons $i$ and $j$.  However, $z = T^{-1}(r)$, so combining Equations (\ref{eq:proof1}) and (\ref{eq:proof2}) gives $z = \pi_{ij}z$.  Plugging this into the equation for the change of variables gives
    \begin{equation}
        \rho_r(r) = \rho_z(T^{-1}(r))|\det J_{T^{-1}}(r)| = \rho_z(z)|\det J_{T^{-1}}(r)|
        \label{eq:proof3}
    \end{equation}
    But we know that $z$ is such that $z = \pi_{ij}z$, which means that $z_i = z_j$; and for such $z$'s, we know that $\rho_z(z) = 0$, by the assumption of condition (R2) for the base density.  Thus, plugging back into Equation (\ref{eq:proof3}) gives
    \begin{equation}
        \rho_r(r) = 0 \cdot |\det J_{T^{-1}}(r)| = 0
    \end{equation}
    as desired.  
\end{proof}

\section{Proof of Theorem \ref{thm:base_density}}
\label{prf:base_density}

\begin{theorem*}
    Let $\rho_z(z) = \dpp(z_u; n_u) \dpp(z_d; n_d)$.  Then $\rho_z$ satisfies conditions (R1) and (R2) from Theorem 2. 
\end{theorem*}

\begin{proof}
    Let us begin with property (R1): we would like to prove that $\rho_z$ is $\G$-invariant.  Let $\pi \in \G$; as $\G = \mathbb{S}_{\nup} \times \mathbb{S}_{\ndn}$, we may write the permutation $\pi = \pi_u \otimes \pi_d$, where $\pi_u$ is a permutation which applies to the indices in $\nup$, and similarly for $\pi_d$ and the indices in $\ndn$.  Thus,
    \begin{equation}
        \rho_z(\pi z) = \dpp(\pi_u z_u; n_u) \dpp(\pi_d z_d; n_d)
    \end{equation}
    Now, recall that the Projection DPP's density is defined by
    \begin{equation}
        \mathbf{K}_n(r) =
        \begin{bmatrix}
            K(r_1, r_1) & \dots & K(r_1, r_n) \\
            \vdots & \ddots & \vdots \\
            K(r_n, r_1) & \dots & K(r_n, r_n)
        \end{bmatrix}
        \qquad \Rightarrow \qquad
        \dpp(r; n) = \frac{1}{n!} \det \mathbf{K}_n(r)
    \end{equation}
    Thus, we must compute $\mathbf{K}_n(\pi r)$ for a permutation $\pi$.  We may represent the action of $\pi$ on a vector of length $n$ by an $n \times n$ matrix $P_\pi$.  It is then straightforward to see that
    \begin{equation}
        \mathbf{K}_n(\pi r) = P_\pi \mathbf{K}_n(r) P_\pi^T
    \end{equation}
    and thus that
    \begin{equation}
        \det \mathbf{K}_n(\pi r) = \det \left( P_\pi \mathbf{K}_n(r) P_\pi^T \right) = \det \left( P_\pi^T P_\pi \mathbf{K}_n(r) \right) = \det(P_\pi)^2 \det \mathbf{K}_n(r) = \det \mathbf{K}_n(r)
    \end{equation}
    where the second equality uses the cyclic property of the determinant; the third equality that a determinant of products is the product of determinants; and the fourth equality that the determinant of a permutation matrix $P_\pi$ is $\pm 1$.  Thus, we have that
    \begin{equation}
        \dpp(\pi_u z_u; n_u) = \frac{1}{n_u !} \det \mathbf{K}_{n_u}(\pi_u z_u) = \frac{1}{n_u !} \det \mathbf{K}_{n_u}(z_u) = \dpp(z_u; n_u)
    \end{equation}
    Likewise, $\dpp(\pi_d z_d; n_d) = \dpp(z_d; n_d)$.  This gives finally that $\rho_z(\pi z) = \rho_z(z)$, establishing that $\rho_z$ is $\G$-invariant, i.e.~satisfies condition (R1).

    We now turn to condition (R2): we would like to prove that $\rho_z(z) = 0 \text{ if } z_i = z_j, \text{ for } i, j \in \nup \text{ or } i, j \in \ndn$.  Let us focus on the case of spin-up electrons, i.e.~$i, j \in \nup$; the spin-down case will follow analogously.  We know that 
    \begin{equation}
        \dpp(z_u; n_u) = \frac{1}{n_u !} \det \mathbf{K}_{n_u}(z_u)
    \end{equation}
    Given the definition of the matrix $\mathbf{K}_{n_u}(z_u)$, it is straightforward to see that if $z_i = z_j$, then $\mathbf{K}_{n_u}(z_u)$ has identical columns for $i$ and $j$.  However, a matrix with two identical columns is rank deficient, and therefore has determinant $0$.  Thus, we have that $\dpp(z_u; n_u) = 0$ so that $\rho_z(z) = 0$, establishing that $\rho_z$ satisfies condition (R2).
\end{proof}

\section{Sampling Procedure for Projection DPPs}
\label{app:sampling_DPPs}

In order to sample from a Projection DPP, we may follows the procedure outlined in \citep{lavancier2015determinantal}, which we reproduce in Algorithm \ref{alg:DPP_sample}.  We note that the speed of the sampling algorithm is largely unimportant, as one may sample as many samples as one would like offline, prior to (and independent from) the process of minimizing the variational objective.
\begin{algorithm}[hbt!]
\caption{Sampling from a Projection Determinantal Point Process}
\label{alg:DPP_sample}
\begin{algorithmic}
    \REQUIRE $n$, $H(y)$
    \STATE sample $r_n$ from the distribution with density $\rho_n(y) = \frac{1}{n} \| H(y) \|^2$
    \STATE $e_1 \gets H(r_n) / \| H(r_n) \|$
    \FOR{$i = n-1$ to $1$}
        \STATE sample $r_i$ from the distribution with density $\rho_i(y) = \frac{1}{i} \| H(y) \|^2 - \frac{1}{i} \sum_{j=1}^{n-i} |e_j^T H(y)|^2$
        \STATE $c_i \gets H(r_i) - \frac{1}{i} \sum_{j=1}^{n-i} \left( e_j^T H(r_i)  \right) e_j$
        \STATE $e_{n-i+1} \gets c_i / \| c_i \|$
    \ENDFOR
    \RETURN $r = (r_1, \dots, r_n)$
\end{algorithmic}
\end{algorithm}

In order to sample from $\rho_i(y)$, one may use rejection sampling; for further details, see \citep{lavancier2015determinantal}.  Note that the algorithm can be generalized in a straightforward fashion to a complex orthonormal basis $H(x)$ by replacing all transposes with Hermitian transposes.

\section{Proof of Theorem \ref{thm:equivariance_by_spin}}
\label{prf:equivariance_by_spin}

\begin{theorem*}
    The transformation $T^\l$ is $\G$-equivariant if and only if 
    \begin{equation*}
        \Tal(\pi_\alpha \ral, \pi_\alphac \rcl) = \pi_\alpha \Tal(\ral, \rcl) \qquad \alpha \in \{u, d\}
    \end{equation*}
    That is, $\Tal$ is equivariant with respect to $r_\alpha^\l$, and invariant with respect to $r_\alphac^\l$.
\end{theorem*}

\begin{proof}
    Let us begin with the forward direction: suppose that $T^\l$ is $\G$-equivariant.
    Let $\pi \in \G$; as $\G = \mathbb{S}_{\nup} \times \mathbb{S}_{\ndn}$, we may write the permutation $\pi = \pi_u \otimes \pi_d$, where $\pi_u$ is a permutation which applies to the indices in $\nup$, and similarly for $\pi_d$ and the indices in $\ndn$.  Then $\G$-equivariance of $T^\l$ implies
    \begin{equation}
        T^\l(\pi r^\l) = \pi T^\l(r^\l) = \pi r^\lp
        \label{eq:layer_equivariance}
    \end{equation}
    Now, let us break this down by spin.  Note that
    \begin{equation}
        \pi r^\lp = (\pi_u r_u^\lp, \pi_d r_d^\lp)
    \end{equation}
    and also
    \begin{equation}
        T^\l(\pi r^\l) = (T_u^\l(\pi_u r_u^\l, \pi_d r_d^\l) \, , \, T_d^\l(\pi_u r_u^\l, \pi_d r_d^\l))
    \end{equation}
    But Equation (\ref{eq:layer_equivariance}) says that $\pi r^\lp = T^\l(\pi r^\l)$, so we may combine the last two equations to give
    \begin{equation}
        T_u^\l(\pi_u r_u^\l, \pi_d r_d^\l) = \pi_u T_u^\l(r_u^\l, r_d^\l) \qquad \text{and} \qquad T_d^\l(\pi_u r_u^\l, \pi_d r_d^\l) = \pi_d T_d^\l(r_u^\l, r_d^\l)
        \label{eq:layer_eq_inv_orig}
    \end{equation}
    In words, $T_u^\l$ is \textit{equivariant} with respect to $\pi_u$, and \textit{invariant} with respect to $\pi_d$; and the reverse is true for $T_d^\l$.
    For notational convenience, we use $\alpha \in \{u, d\}$ to denote the spin, and the complement of the spin is given by $\alphac$ (i.e.~if $\alpha = u$ then $\alphac=d$).  In this case, we may summarize Equation (\ref{eq:layer_eq_inv_orig}) as
    \begin{equation}
        \Tal(\pi_\alpha r_\alpha^\l, \pi_\alphac r_\alphac^\l) = \pi_\alpha \Tal(r_\alpha^\l, r_\alphac^\l) \qquad \alpha \in \{u, d\}
    \end{equation}
    which completes the proof for the forward direction.

    Now, suppose that $\Tal(\pi_\alpha \ral, \pi_\alphac \rcl) = \pi_\alpha \Tal(\ral, \rcl)$.  For a given permutation $\pi = \pi_u \otimes \pi_d$, we have
    \begin{equation}
        T^\l(\pi r^\l) = (T_u^\l(\pi_u r_u^\l, \pi_d r_d^\l) \, , \, T_d^\l(\pi_u r_u^\l, \pi_d r_d^\l))
                       = (\pi_u T_u^\l(r_u^\l, r_d^\l) \, , \, \pi_d  T_d^\l(r_u^\l, r_d^\l))
                       = \pi T^l(r^\l)
    \end{equation}
    so that $T^\l$ is $\G$-equivariant, as desired.  This completes the proof for the reverse direction.
\end{proof}

\section{Proof of Theorem \ref{thm:continuous_flow}}
\label{prf:continuous_flow}

\begin{theorem*}
    Let the transformation $r = T(z)$ be specified by the ODE
    \begin{equation*}
        \frac{dv}{dt} = \Gamma(v), \quad \text{with} \quad  v(0) = z \sim \rho_z(\cdot) \quad \text{and} \quad r = v(1)
    \end{equation*}
    Then $T$ is $\G$-equivariant if $\Gamma$ is $\G$-equivariant.
\end{theorem*}

\begin{proof}
    The result follows directly from Theorem 2 in \citep{kohler2020equivariant}.
\end{proof}

\section{The Complexity of Discrete Normalizing Flows}
\label{app:complexity}

The limiting factor in the complexity of the discrete normalizing flow is the computation of determinants; unlike traditional normalizing flows, we make no effort to accelerate the determinant of the Jacobian, which allows us to have more expressive layers.  In particular, the relevant space is of dimension $Dn$, where $D=3$ and $n$ is on the order of tens of electrons for small molecules.  Thus, the overall dimension of the space is low hundreds.

We note that the determinant of the Jacobian is cubic in the dimension; for a low-dimensional space this is acceptable.  Furthermore, popular methods based on neural networks, such as FermiNet \citep{pfau2020ab} and PauliNet \citep{hermann2020deep} use determinants explicitly in their \ansatze{}, so that they have similar complexity.  However, these methods use Markov Chain Monte Carlo sampling, so that they incur extra overhead from having to sample by solving for the limit of a stochastic differential equation, which our method avoids.

\section{Proof of Theorem \ref{thm:layer_inverse_equivariance}}
\label{prf:layer_inverse_equivariance}

\begin{theorem*}
    Let $T^\l$ be a Split Subspace Layer, as given by
    \begin{equation*}
        \railp = \Tail(\ral, \rcl) = \rail + \xal \Pail(\gal, \gcl) \quad \text{with} \quad \Pail(\gal, \gcl) \in \mathbb{R}^{D-D_\beta}
    \end{equation*}
    Then $T^\l$ is invertible.  In particular, let $\gailrp = (\bal)^T \railp$; then the inverse of the layer is given by
    \begin{equation*}
        \rail = \railp - \xal \Pail(\galrp, \gclrp)
    \end{equation*}
    Furthermore, the layer $T^\l$ is $\G$-equivariant if
    \begin{equation*}
        \Pal(\pi_\alpha \gal, \pi_\alphac \gcl) = \pi_\alpha \Pal(\gal, \gcl)
    \end{equation*}
    i.e.~if $\Pal(\gal, \gcl)$ is equivariant with respect to permutations on $\gal$ and invariant with respect to permutations on $\gcl$.
\end{theorem*}

\begin{proof}
    Let us first prove the layer's inverse.  First, note that $\gailrp$ can be computed entirely from variables in layer $\lp$.  Also note that $\gailrp \neq \gailp$, since $\gailp = (\balp)^T \railp$ - i.e. $\gailrp$ uses $\bal$, while $\gailp$ uses $\balp$.  Now, we show that
    \begin{align*}
        \gailrp & = (\bal)^T \railp \\
                & = (\bal)^T \left( \rail + \xal \Pail(\gal, \gcl) \right) \\
                & = (\bal)^T \rail + (\bal)^T \xal \Pail(\gal, \gcl) \\
                & = \gail + 0  = \gail
    \end{align*}
    where the equality in the last line follows from the fact that $\xal$ is the orthogonal complement of $\bal$, so that $(\bal)^T \xal = 0$.  That is, the Split Subspace Layer has the nice property that it preserves projections onto the subspace given by $\bal$.

    Recall that the Split Subspace Layer $T_\alpha^\l$ on a per-electron basis by
    \begin{equation}
        \railp = \Tail(\ral, \rcl) = \rail + \xal \Pail(\gal, \gcl)
        \label{eq_layer_main_supp}
    \end{equation}
    Given this, the inverse is straightforwardly computed by rearranging Equation (\ref{eq_layer_main_supp}):
    \begin{equation}
        \rail = \railp - \xal \Pail(\gal, \gcl) = \railp - \xal \Pail(\galrp, \gclrp)
    \end{equation}
    Note that everything on the right-hand side depends on variables from layer $\lp$, as desired.  Thus, we have shown the layer in Equation (\ref{eq_layer_main_supp}) is invertible regardless of the form of the network $\Pal$.
    
    Now let us turn to proving the layer's $\G$-equivariance.  Recall that the conditions for $\G$-equivariance are given by
    \begin{equation}
        \Tal(\pi_\alpha \ral, \pi_\alphac \rcl) = \pi_\alpha \Tal(\ral, \rcl) \qquad \alpha \in \{u, d\}
    \end{equation}
    This can be combined with Equation (\ref{eq_layer_main_supp}):
    \begin{align}
        & \Tal(\pi_\alpha \ral, \pi_\alphac \rcl) = \pi_\alpha \Tal(\ral, \rcl) \notag \\
        \Leftrightarrow \quad & \Tail(\pi_\alpha \ral, \pi_\alphac \rcl) = T_{\alpha, \pi_{\alpha}(i)}^\l (\ral, \rcl) \notag \\
        \Leftrightarrow \quad & r_{\alpha, \pi_{\alpha}(i)}^\l + \xal \Pail(\pi_\alpha \gal, \pi_\alpha \gcl) = r_{\alpha, \pi_{\alpha}(i)}^\l + \xal \varphi_{\alpha, \pi_{\alpha}(i)}^\l(\gal, \gcl) \notag \\
        \Leftrightarrow \quad & \Pail(\pi_\alpha \gal, \pi_\alpha \gcl) = \varphi_{\alpha, \pi_{\alpha}(i)}^\l(\gal, \gcl) \notag \\
        \Leftrightarrow \quad & \Pal(\pi_\alpha \gal, \pi_\alphac \gcl) = \pi_\alpha \Pal(\gal, \gcl)
    \end{align}
    where $\pi_{\alpha}(i)$ indicates the index that electron $i$ is moved to under the permutation $\pi_{\alpha}$; and the fourth line follows from the fact that the previous statement must be true for all possible outputs of $\Pal$.  This completes the proof.
\end{proof}

\section[Implementation of the G-Equivariant Layer]{Implementation of the $\G$-Equivariant Layer}
\label{app:implementation}

As we have seen, invertibility places no special restrictions on the form of $\Pal$.  With regard to the conditions imposed by $\G$-equivariance, i.e.~
\begin{equation}
    \Pal(\pi_\alpha \gal, \pi_\alphac \gcl) = \pi_\alpha \Pal(\gal, \gcl),
    \label{eq:phi_equivariance_supp}
\end{equation}
there are several ways to achieve them.  We propose the following method, as it uses standard off-the-shelf architectures; we use the variables $\zail$ to represent intermediate quantities.
\begin{enumerate}
    \item \textit{\textbf{Lifting:}} Map each value $\gail$ from dimension $D_\beta$ to dimension $D_\zeta$:
    \begin{equation}
        \zail = W_\alpha \gail
    \end{equation}
    where there are two matrices $W_\alpha$ of dimension $D_\zeta \times D_\beta$, one for each spin $\alpha \in \{u, d\}$.
    \item \textit{\textbf{Multihead Attention:}} We have two Multihead Attention (MHA) layers $\tal$, one for each spin.  Each MHA takes as input the the list $\zal = \{ \zail \}_{i \in \na}$.  The output of the MHA is then 
    \begin{equation}
        \zal \leftarrow \tal(\zal)
    \end{equation}
    \item \textit{\textbf{Fully Connected Layer Per Spin:}} There are two fully connected layers $\mal$, one for each spin.  The layer is applied per electron, with the same layer being applied to electrons of a given spin:
    \begin{equation}
        \zail \leftarrow \mal(\zail) \quad \text{for } i \in \na
    \end{equation}
    \item \textit{\textbf{Average:}} Form the average values: $\bzal = \frac{1}{\numa} \sum_{i \in \na} \zail$.
    \item \textit{\textbf{Fully Connected Layer with Spin Mixing:}} We have two fully connected layers $\hmal$, one for each spin.  Then:
    \begin{equation}
        \Pail(\gal, \gcl) = \hmal(\texttt{CAT}(\zail \, , \, \bzcl)) \quad \text{for } i \in \na
        \label{eq:mlp}
    \end{equation}
    The output of the MLPs $\mal$ is of dimension $D-D_\beta$.
\end{enumerate}
Due to the permutation-equivariance of Multihead Attention, the $\G$-equivariance follows naturally.  Some comments are in order:
\begin{itemize}
    \item  We can choose $D_\beta \in \{ 1, \dots, D-1 \}$.  Since in our case $D = 3$, this gives us exactly two choices: $D_\beta = 1$ or $D_\beta = 2$.
    \item The fully connected layers should use smooth activation functions, i.e.~not ReLU.  There are many possible smooth substitutes for ReLU-like activations, such as Swish, SiLU, etc.
    \item To achieve orthogonalization, i.e. to ensure that $\xal$ is itself orthonormal and is also orthogonal to $\bal$, it is important to use a smooth procedure. Gram-Schmidt may be employed for this purpose: an initial (e.g.~random) set of vectors are chosen, which are then orthonormalized by the procedure.
    \item In the special case of Helium, there are only 2 electrons: one which is spin-up, and the other which is spin-down.  In this case, the requirement that $\Pal(\gal, \gcl)$ be equivariant with respect to permutations of $\gal$ is trivially satisfied; likewise, the requirement that $\Pal(\gal, \gcl)$ be invariant with respect to permutations of $\gcl$ is also trivially satisfied.  As a result, the Multihead Attention layers $\tal$ may be replaced by the identity, with everything else remaining the same. 
\end{itemize}

\section{A Generalized Variant of the Split Subspace Layer}
\label{app:generalized_variant}

We note that a generalized variant of the Split Subspace Layer is as follows:
\begin{equation}
    \railp = \Tail(\ral, \rcl) = \bal \eta_{\alpha, i}^\l(\gal, \gcl) + \xal \Pail(\gal, \gcl)
\end{equation}
where both $\phi$ and $\eta$ satisfy the conditions in (\ref{eq:phi_equivariance_supp}), and $\eta$ is explicitly invertible in the sense that the system of equations $y_{\alpha, i} = \eta_{\alpha, i}^\l(\gal, \gcl)$ for all $\alpha, i$ may be inverted to solve for all values of $\gail$.  An example of such an $\eta$ is given by $\eta_{\alpha, i}^\l(\gal, \gcl) = f(A \gail + \sum_{j \neq i} B \gamma_{\alpha,j}^\l + \sum_j C \gamma_{\alphac,j}^\l)$ for $D_\beta \times D_\beta$ matrices $A, B, C$ and an invertible nonlinearity $f : \mathbb{R}^{D_\beta} \to \mathbb{R}^{D_\beta}$ (such as the cube of each element).

\section{Proof of Theorem \ref{thm:log_domain}}
\label{prf:log_domain}

\begin{theorem*}
    The local energy can be written as
    \begin{equation*}
        \elocal_r(r; \theta) = -\tfrac{1}{2} \Delta_r q(r; \theta) -\tfrac{1}{2} \| \nabla_r q(r; \theta) \|^2 + V(r)
    \end{equation*}
    In particular, the local energy is independent of the phase $w(r; \theta)$.  Furthermore, let
    \begin{equation*}
        \Omega(r; \theta) = 2 \left( \elocal_r(r; \theta) - \mathbb{E}_{r \sim \rho(\cdot; \theta)} \left[ \elocal_r(r; \theta) \right] \right) \nabla_\theta q(r; \theta)
    \end{equation*}
    Then the gradient of the loss function may be written as
    \begin{equation*}
        \nabla_\theta \mathcal{L}(\theta) = \mathbb{E}_{r \sim \rho(\cdot; \theta)} \left[ \Omega(r; \theta) \right]
                                          \approx \frac{1}{K} \sum_{k=1}^K \Omega \left(r^{(k)}; \theta \right)
    \end{equation*}
    with samples $r^{(k)} \sim \rho(\cdot; \theta)$.
\end{theorem*}

\begin{proof}
    For the moment, we suppress $\theta$ for convenience.  Recall that the overall (i.e.~complex) local energy is defined by
    \begin{equation}
        \elocal(r) = \frac{H \psi(r)}{\psi(r)} = -\frac{\Delta \psi(r)}{2\psi(r)} + V(r)
    \end{equation}
    Let $r_{j, d}$ be the $d^{th}$ component of the position vector of the $j^{th}$ electron, $d = 1, \dots, D$.  Then plugging in $\psi(r) = e^{q(r) + iw(r)}$, we have that
    \begin{equation}
        \frac{\partial \psi}{\partial r_{j,d}}
        = e^{q(r) + iw(r)}\left( \frac{\partial q}{\partial r_{j,d}} + i \frac{\partial w}{\partial r_{j,d}} \right)
        = \psi(r) \left( \frac{\partial q}{\partial r_{j,d}} + i \frac{\partial w}{\partial r_{j,d}} \right)
    \end{equation}
    and
    \begin{align}
        \frac{\partial^2 \psi}{\partial r_{j,d}^2} & = e^{q(r) + iw(r)}\left( \frac{\partial^2 q}{\partial r_{j,d}^2} + i \frac{\partial^2 w}{\partial r_{j,d}^2} \right) +  e^{q(r) + iw(r)} \left( \frac{\partial q}{\partial r_{j,d}} + i \frac{\partial w}{\partial r_{j,d}} \right)^2 \notag \\
        & = \psi(r) \left[ \left( \frac{\partial^2 q}{\partial r_{j,d}^2} + \left( \frac{\partial q}{\partial r_{j,d}} \right)^2 - \left( \frac{\partial w}{\partial r_{j,d}} \right)^2 \right) + i \left( \frac{\partial^2 w}{\partial r_{j,d}^2} + 2  \frac{\partial q}{\partial r_{j,d}} \frac{\partial w}{\partial r_{j,d}}\right) \right]
    \end{align}
    With the appropriate summation, this immediately yields
    \begin{equation}
        \elocal(r) = -\frac{1}{2}\sum_{j=1}^n \left[ \left( \Delta_j q + \| \nabla_j q \|^2 - \| \nabla_j w \|^2 \right) + i \left( \Delta_j w + 2 \nabla_j q \cdot \nabla_j w \right) \right] + V(r)
    \end{equation}
    so that its real part simplifies to
    \begin{align}
        \elocal_r(r) 
        & = -\frac{1}{2}\sum_{j=1}^n \left( \| \nabla_j w \|^2 - \Delta_j q - \| \nabla_j q \|^2 \right) + V(r) \notag \\
        & = \tfrac{1}{2} \left( \| \nabla w \|^2 - \Delta q - \| \nabla q \|^2 \right) + V(r) 
    \end{align}

    Now, it is known that since the Hamiltonian is time-reversal invariant  and  Hermitian, both its eigenvalues and its eigenfunctions are real.  Since the ground-state wavefunction we are looking for is real, the phase $w(r)$ can be taken to belong to the two element set $\{0, \pi\}$, where $w(r) = 0$ corresponds to positive values of the wavefunction $\psi(r)$, and $w(r) = \pi$ to negative values of $\psi(r)$.  Thus, where the sign of $\psi(r)$ does not change $w(r)$ is constant, and therefore $\| \nabla w (r) \|^2 = 0$.

    We are then left to consider the case when the sign of $\psi(r)$ flips, and therefore there is a discontinuity in $w(r)$; this occurs precisely where $\psi(r) = 0$.  However, recall that the loss is given by
    \begin{equation}
        \mathcal{L}(\theta) = 
        \langle \psi(\cdot; \theta) | H | \psi(\cdot; \theta) \rangle
        \, = \, \mathbb{E}_{r \sim \rho(\cdot; \theta)} \left[ \elocal_r(r; \theta) \right]
    \end{equation}
    When $\psi(r) = 0$ then $\rho(r) = 0$; thus, samples where there is a discontinuity are never selected.  We may therefore set the local energy at such values of $r$ to any value we wish, without affecting the value of $\mathcal{L}(\theta)$.  In particular, we are free to set $\| \nabla w(r) \|= 0$ at such points.  In conclusion, then, we have demonstrated that
    \begin{equation}
        \elocal_r(r; \theta) = -\tfrac{1}{2} \Delta_r q(r; \theta) -\tfrac{1}{2} \| \nabla_r q(r; \theta) \|^2 + V(r)
    \end{equation}
    which is independent of the phase $w(r)$.

    Turning to the second part of the theorem, we note that
    \begin{equation}
        \mathcal{L}(\theta)
        \, = \, \mathbb{E}_{r \sim \rho(\cdot; \theta)} \left[ \elocal_r(r; \theta) \right]
        \, = \, \int \elocal_r(r; \theta) \rho(r; \theta) dr
    \end{equation}
    so that
    \begin{equation}
        \nabla_\theta \mathcal{L}(\theta)
        \, = \, \int \left( \nabla_\theta \elocal_r(r; \theta) \rho(r; \theta) + \elocal_r(r; \theta) \nabla_\theta \rho(r; \theta) \right) dr
    \end{equation}
    However, since $q(r; \theta) = \frac{1}{2} \log \rho(r; \theta)$, then $\nabla_\theta q(r; \theta) = \nabla_\theta \rho(r; \theta) / 2 \rho(r; \theta)$, or $\nabla_\theta \rho(r; \theta) = 2 \rho(r; \theta) \nabla_\theta q(r; \theta)$.
    Plugging this in gives
    \begin{align}
        \nabla_\theta \mathcal{L}(\theta)
        & \, = \, \int \left( \nabla_\theta \elocal_r(r; \theta) \rho(r; \theta) + 2 \elocal_r(r; \theta) \rho(r; \theta) \nabla_\theta q(r; \theta) \right) dr \notag \\
        & \, = \, \int \rho(r; \theta) \left( \nabla_\theta \elocal_r(r; \theta) + 2 \elocal_r(r; \theta) \nabla_\theta q(r; \theta) \right) dr \label{eq:grad_loss}
    \end{align}
    Now, let us examine the vector $\zeta \equiv \int \rho(r; \theta) \nabla_\theta \elocal_r(r; \theta) dr$.  First, note that as the eigenfunctions are real, then $\elocal_r(r; \theta) = \elocal(r; \theta)$.  Then, a straightforward computation gives that the $j^{th}$ component of $\nabla_\theta \elocal(r; \theta)$ is
    \begin{equation}
        \partial_{\theta_j} \elocal(r; \theta)
        = \left[ \partial_{\theta_j} \frac{H\psi(r; \theta)}{\psi(r; \theta)} \right]
        = \frac{\psi(r; \theta) \, (\partial_{\theta_j} H \psi)(r; \theta) - (H\psi)(r; \theta) \,\, (\partial_{\theta_j}\psi)(r; \theta)}{\psi^2(r; \theta)}
    \end{equation}
    We may interchange the order of differentiation so that
    \begin{equation}
        \partial_{\theta_j} \elocal(r; \theta)
        = \frac{\psi(r; \theta) \, (H \partial_{\theta_j} \psi)(r; \theta) - (H\psi)(r; \theta) \,\, (\partial_{\theta_j}\psi)(r; \theta)}{\psi^2(r; \theta)}
    \end{equation}
    Thus, the $j^{th}$ component of $\zeta$ is given by
    \begin{equation}
        \zeta_j = \int \left( \psi(r; \theta) \, (H\partial_{\theta_j} \psi)(r; \theta) - (H\psi)(r; \theta) \,\, (\partial_{\theta_j} \psi)(r; \theta) \right) dr
    \end{equation}
    where we have used $\rho = |\psi|^2 = \psi^2$ since the eigenfunctions are real.  We may rewrite this as
    \begin{equation}
        \zeta_j = \langle \psi , H (\partial_{\theta_j} \psi) \rangle - \langle H\psi , \partial_{\theta_j} \psi \rangle
    \end{equation}
    However, note that
    \begin{equation}
        \langle \psi , H (\partial_{\theta_j} \psi) \rangle
        = \langle H^\dagger \psi , \partial_{\theta_j} \psi \rangle
        = \langle H \psi , \partial_{\theta_j} \psi \rangle
    \end{equation}
    where the latter equality follows from the Hermiticity of the Hamiltonian $H$.  Thus, we have shown that $\zeta_j = 0$; as this is true for any $j$, we have that $\zeta = 0$.  Plugging this into Equation (\ref{eq:grad_loss}) gives 
    \begin{equation}
        \nabla_\theta \mathcal{L}(\theta)
        \, = \, 2 \int \rho(r; \theta) \elocal_r(r; \theta) \nabla_\theta q(r; \theta) dr
        \label{eq:nearly}
    \end{equation}
    Finally, using the Expected Grad-Log-Prob Lemma, we have that
    \begin{equation}
        \mathbb{E}_{r \sim \rho(\cdot; \theta)} \left[ \nabla_\theta q(r; \theta) \right]
        = \int \rho(r; \theta) \nabla_\theta q(r; \theta) dr = 0
    \end{equation}
    Multiplying the above equation by $\mathbb{E}_{r \sim \rho(\cdot; \theta)} \left[ \elocal_r(r; \theta) \right]$ and subtracting off from Equation (\ref{eq:nearly}) yields
    \begin{align}
        \nabla_\theta \mathcal{L}(\theta)
        & \, = \, 2 \int \rho(r; \theta) \left( \elocal_r(r; \theta) - \mathbb{E}_{r \sim \rho(\cdot; \theta)} \left[ \elocal_r(r; \theta) \right] \right) \nabla_\theta q(r; \theta) dr \notag \\
        & \, = \, \mathbb{E}_{r \sim \rho(\cdot; \theta)} \left[ \Omega(r; \theta) \right]
    \end{align}
    where
    \begin{equation}
        \Omega(r; \theta) = 2 \left( \elocal_r(r; \theta) - \mathbb{E}_{r \sim \rho(\cdot; \theta)} \left[ \elocal_r(r; \theta) \right] \right) \nabla_\theta q(r; \theta)
    \end{equation}
    as desired.
\end{proof}

\section{Optimization of the Objective Function}
\label{app:optimization}

The sampled approximation to the objective is given by
\begin{equation}
    \mathcal{L}(\theta)
    \, \approx \,
    \frac{1}{K} \sum_{k=1}^K \elocal_r\left(r^{(k)}; \theta \right)
\end{equation}
In order to optimize this objective, we use the procedure in Algorithm \ref{alg:ground_state}; first, we recall the following equations:
\begin{equation}
    q(r; \theta) = q_z \left( T^{-1} (r; \theta) \right) + \frac{1}{2} \sum_{\l=0}^L \log \left| \det J_{(T^\l)^{-1}}(r^\lp; \theta) \right|
    \label{eq:q_supp}
\end{equation}
\begin{equation}
    \elocal_r(r; \theta) = -\tfrac{1}{2} \Delta_r q(r; \theta) -\tfrac{1}{2} \| \nabla_r q(r; \theta) \|^2 + V(r)
    \label{eq:elocal_supp}
\end{equation}
\begin{equation}
    \Omega(r; \theta) = \nabla_\theta \elocal_r(r; \theta) + 2 \elocal_r(r; \theta) \nabla_\theta q(r; \theta)
    \label{eq:omega_supp}
\end{equation}
Algorithm \ref{alg:ground_state} is specified for the discrete normalizing flow; the procedure for the continuous normalizing flow will be similar.  Note that we initially sample a large number $K_{large}$ of samples from the base density; we emphasize that this step can be performed entirely offline, and does not entail additional computational complexity.
\begin{algorithm}[hbt!]
\caption{Computation of Ground State Wavefunction and Energy}
\label{alg:ground_state}
\begin{algorithmic}
    \REQUIRE base log-density $q_z(\cdot)$, normalizing flow $\{T^\l(\cdot; \theta)\}_{\l=0}^L$, potential $V(\cdot)$, learning rate $\epsilon$
    \STATE sample $\mathcal{Z} = \left\{ z^{(k)} \right\}_{k=1}^{K_{large}}$ for $z^{(k)} \sim \rho_z(\cdot)$ and $K_{large}$ a very large number of samples
    \STATE take $q(r; \theta)$ from (\ref{eq:q_supp}) and use auto-differentiation to compute $\nabla_r q(r; \theta)$ and $\Delta_r q(r; \theta)$
    \STATE using $\nabla_r q(r; \theta)$ and $\Delta_r q(r; \theta)$, compute $\elocal_r(r; \theta)$ from (\ref{eq:elocal_supp})
    \STATE using auto-differentiation, compute the function $\Omega(r; \theta)$ as in (\ref{eq:omega_supp})
    \STATE initialize $\theta$, e.g.~using Xavier initialization
    \WHILE{not converged}
        \STATE sample $K$ values of $z^{(k)}$ from $\mathcal{Z}$
        \STATE compute $r^{(k)} = T(z^{(k)}; \theta)$ using $T = T_L \circ \dots \circ T_0$
        \STATE compute the energy $E = \frac{1}{K} \sum_{k=1}^K \elocal_r(r^{(k)}; \theta)$
        \STATE compute the gradient $g = \frac{1}{K} \sum_{k=1}^K \Omega(r^{(k)}; \theta)$
        \STATE take $\theta \leftarrow \theta - \epsilon g$
    \ENDWHILE
    \RETURN $E$, $\theta$
\end{algorithmic}
\end{algorithm}

\section{Proof of Theorem \ref{thm:phase}}
\label{prf:phase}

\begin{theorem*}
    Let $\rho_0(r)$ be the density for the ground state wavefunction. Let $\prec$ be a strict total order on $\mathbb{R}^D$, and define the set
    \begin{equation*}
        \mathcal{R} = \{r = (r_1, \dots r_n): r_1 \prec r_2 \prec \dots \prec r_{n_u} \,\, \text{ and } \,\, r_{n_u+1} \prec r_{n_u+2} \prec \dots \prec r_n \}
    \end{equation*}
    For any $r$ without $r_i = r_j$, define the permutation $\pir \in \G$ by $\pir r \in \mathcal{R}$.  Then a valid antisymmetric ground state wavefunction is given by
    \begin{equation*}
        \psi_0(r) = 
        \begin{cases}
            (-1)^\pir \sqrt{\rho_0(r)} & \text{if } r_i \neq r_j \,\, \forall i, j \\
            0 & \text{otherwise}
        \end{cases}
    \end{equation*}
\end{theorem*}

\begin{proof}
    We begin by noting that the set $\mathcal{R}$ contains the spin-up electrons in ascending order, according to the ordering relation $\prec$, and the spin-down electrons also in ascending order.  Now, begin by considering the case of $r$ for which $r_i = r_j$ for some pair of electrons $i$ and $j$; in this case, $\psi_0(r) = 0$, as is required by antisymmetry.  Now, consider the case of $r$ for which $r_i \neq r_j \,\, \forall i, j$.  In this case, for any permutation $\pi \in \G$ we have that
    \begin{equation}
        \psi_0(\pi r) = (-1)^{\pi_\prec(\pi r)} \sqrt{\rho_0(\pi r)}
        \label{eq:psi_tr}
    \end{equation}
    However, recall that $\pir$ is defined by 
    \begin{equation}
        \pir r \in \mathcal{R}
    \end{equation}
    Therefore, $\pi_\prec(\pi r)$ is defined by 
    \begin{equation}
        \pi_\prec(\pi r) \pi r \in \mathcal{R}
    \end{equation}
    Comparing the latter two equations, we see that 
    \begin{equation}
        \pi_\prec(\pi r) \pi = \pir \qquad \Rightarrow \qquad \pi_\prec(\pi r) = \pir \pi^{-1}
        \label{eq:pir}
    \end{equation}
    Furthermore, we know that as $\rho_0(x)$ is the density for the ground state wavefunction, it must satisfy property (D1) of Theorem 2, 
    namely it must be $\G$-invariant; therefore, we must have that
    \begin{equation}
        \rho_0(\pi r) = \rho_0(r)
        \label{eq:rho_as}
    \end{equation}
    Plugging Equations (\ref{eq:pir}) and (\ref{eq:rho_as}) into (\ref{eq:psi_tr}) gives
    \begin{align}
        \psi_0(\pi r) & = (-1)^{\pir \pi^{-1}} \sqrt{\rho_0(r)} \notag \\
                      & = (-1)^\pi (-1)^\pir \sqrt{\rho_0(r)} \notag \\
                      & = (-1)^\pi \psi_0(r) \label{eq:psi_as}
    \end{align}
    where in the second line, we have used the facts that $(-1)^{\pi_a \pi_b} = (-1)^{\pi_a} (-1)^{\pi_b}$; and that $(-1)^{\pi^{-1}} = (-1)^\pi$.  But Equation (\ref{eq:psi_as}) is exactly the antisymmetry property we desire, and so we have completed the proof.

    Finally, we note that $\psi_0(r) > 0$ for $r \in \mathcal{R}$; this is an arbitrary choice, and we could have equally well defined a second ground state wavefunction $\tilde{\psi}_0$ with $\tilde{\psi}_0(r) < 0$ for $r \in \mathcal{R}$.  It is easy to see that in this case, $\tilde{\psi}_0(r) = -\psi_0(r)$ for all $r$.  However, this is not surprising: either $\psi_0$ or $-\psi_0$ may be taken as an eigenfunction of $H$, as eigenfunctions are only defined up to sign.
\end{proof}

\section{Proof of Theorem \ref{thm:cusps}}
\label{prf:cusps}

\begin{theorem*}
    Let the set of distances be given by $\Dl = \left\{ \Dijl \right\}_{i < j}$ where $\Dijl = \| \ril - \rjl \|$.  Given a layer of the form
    \begin{equation*}
        \rilp = \Rl(\Dl; \theta) \, \ril + \tl(\Dl; \theta) \qquad \text{with } \Rl(\Dl; \theta) \in O(D) \text{ and } \tl(\Dl; \theta) \in \mathbb{R}^D
    \end{equation*}
    Then the layer is both $\G$-equivariant as well as invertible.
\end{theorem*}

\begin{proof}
    Let us begin with invertibility.  We may compute the inter-electron distances at layer $\lp$:
    \begin{align}
        \Dijlp & = \| \rilp - \rjlp \| \notag \\
               & = \| \Rl(\Dl; \theta) \, \ril + \tl(\Dl; \theta) - \Rl(\Dl; \theta) \, \rjl - \tl(\Dl; \theta) \| \notag \\
               & = \| \Rl(\Dl; \theta) ( \ril - \rjl ) \| = \| \ril - \rjl \| = \Dijl
    \end{align}
    where the third line holds since $\Rl(\Dl; \theta) \in O(D)$.  That is, since we are rotating and translating all of the electrons with the same rotation matrix and translation vector the inter-electron distances are preserved.  As a result, the inverse is simply
    \begin{align}
        \ril & = \Rl(\Dl; \theta)^{-1} \, (\rilp - \tl(\Dl; \theta)) \notag \\
             & = \Rl(\Dlp; \theta)^T \, (\rilp - \tl(\Dlp; \theta))
    \end{align}
    where we have used the fact that for a rotation matrix, $\Theta^{-1} = \Theta^T$.  Note that all of the arguments on the right-hand side of the equation depend only on quantities from layer $\lp$, as desired.

    Having established invertibility, let us turn to $\G$-equivariance.  Let $\pi \in \G$, and denote the layer by $r^\lp = Q(r^\l)$, so that $\rilp = Q_i(r^\l)$.  Note that since $\Dl$ is the \textit{set} of distances, we have that $\pi \Dl = \Dl$: a set is inherently unordered, and therefore is unaffected by permutations. Then we have that
    \begin{align}
        Q_i(\pi r^\l) & = \Rl(\pi \Dl; \theta) \, r_{\pi(i)}^\l + \tl(\pi \Dl; \theta) \notag \\
                      & = \Rl(\Dl; \theta) \, r_{\pi(i)}^\l + \tl(\Dl; \theta) \notag \\
                      & = Q_{\pi(i)} (r^\l)
    \end{align}
    so that $Q(\pi r^l) = \pi Q(r^\l)$, as desired.
\end{proof}

\section{Implementation of the Electron-Electron Cusp Layer}
\label{app:implementation_cusps}

Recall that
\begin{equation}
    \rilp = \Rl(\Dl; \theta) \, \ril + \tl(\Dl; \theta)
\end{equation}
Therefore, the network must be a function of the \textit{set} of inter-electron distances $\Dl$.  Using multihead attention will be inefficient, as we must apply it to all \textit{pairs} of electrons, leading to quartic complexity.  Instead, we propose the following Deep Set \citep{zaheer2017deep} style layer:
\begin{enumerate}
    \item \textit{\textbf{MLP Per Electron Pair:}} Apply the same Multilayer Perceptron $\el$ to each electron pair individually:
    \begin{equation}
        \zijl = \el(\Dijl) \quad \text{for all } i < j
    \end{equation}
    \item \textit{\textbf{Average:}} Form the average value: $\bzl = \frac{1}{\frac{1}{2}n(n-1)} \sum_{i < j} \zijl$.
    \item \textit{\textbf{Overall MLP:}} Apply a Multilayer Perceptron $\hel$ to the average:
    \begin{equation}
        \bzl \leftarrow \hel(\bzl)
    \end{equation}
    The output should be of dimension $D^2 + D$, which is equal to $12$ when $D=3$.
    \item \textit{\textbf{Split into Rotation and Translation:}}
    \begin{align}
        \tl & = \text{First $D$ components of } \bzl \notag \\
        A^\l & = \text{Last $D^2$ components of } \bzl, \text{ reshaped into a $D \times D$ matrix} \\
        B^\l & = A^\l - (A^\l)^T, \text{ a skew-symmetric matrix} \notag \\
        \Rl & = \exp(B^\l), \text{ using the matrix exponential} \notag
    \end{align}
\end{enumerate}

\noindent \textbf{Notes}:
\begin{itemize}
    \item The reason we parameterize the rotation as an exponential of a skew-symmetric matrix is so that the layer can effectively be a residual-style layer: if we choose $A^\l = 0$ and $\tl = 0$, then we recover $\rilp = \ril$.  (This is harder if we use a rotation matrix directly, as the identity transformation $\rilp = \ril$ is only recovered if $\Rl = I$, which is harder to achieve.)
    \item It is proposed to use one such layer, or a very small number of such layers, somewhere near the beginning of the flow.  The work of incorporating the cusps in the appropriate manner can then be performed by subsequent layers.
\end{itemize}

\section{Electron-Nuclear Cusps}
\label{app:electron_nuclear_cusps}

It is also known that the gradient of the wavefunction should exhibit a discontinuity when an electron and nucleus coincide.  As in the case of electron-electron cusps, we may treat this by incorporating the electron-nuclear distances directly; we may design our layer exactly analogously to the electron-electron cusp layer, with one main caveat: to preserve invertibility, we can only deal with a single nucleus at a time.  In particular, for a given nucleus $I$ with position $R_I$, let $\DIl = \left\{ \DiIl \right\}_{i=1}^n$ with $\DiIl = \| \ril - R_I \|$.  Then the layer looks like
\begin{equation}
    \rilp = \Rl(\DIl; \theta) \, (\ril - R_I) + R_I
\end{equation}
Note in the above that only the rotation matrix is parameterized, and the translation vector is fixed.  We must include one such layer for each nucleus $I$.

\section{Proof of Theorem \ref{thm:induction}}
\label{prf:induction}

We begin by recalling the equations for nuclear permutation invariance and joint rotation invariance:
\begin{equation}
    \rho(r; \pi R, \pi Z) = \rho(r; R, Z) \quad \text{for } \pi \in \mathbb{S}_N
    \label{eq:nuclear_permutation_invariance_supp}
\end{equation}
\begin{equation}
    \rho(\Theta r; \Theta R, Z) = \rho(r; R, Z) \quad \text{for } \Theta \in O(D)
    \label{eq:joint_rotation_invariance_supp}
\end{equation}
We now state and prove the theorem.
\begin{theorem*}
    Let $\bar{R} = \frac{1}{N}\sum_{I=1}^N R_I = 0$.  Given a continuous normalizing flow of the form $dv/dt = \Gamma_t(v; R, Z)$ with $v(0) = z \sim \rho_z(\cdot)$ and $r = v(1)$.  Let the function $\Gamma_t$ be invariant with respect to nuclear permutations and equivariant with respect to joint rotations, i.e. for all $t$
    \begin{align}
        & \Gamma_t(v; \pi R, \pi Z) = \Gamma_t(v; R, Z) \hspace{2.0mm} \forall \pi \in \mathbb{S}_N \notag \\
        & \Gamma_t(\Theta v; \Theta R, Z) = \Theta \Gamma_t(v; R, Z) \hspace{2.0mm} \forall \Theta \in O(D)
        \label{eq:gamma_extra_properties_supp}
    \end{align}
    Furthermore, suppose that the base density is invariant with respect to rotations, $\rho_z(\Theta z) = \rho_z(z)$ for $\Theta \in O(D)$.  Then the resulting density $\rho(r; R, Z)$ satisfies Equations (\ref{eq:nuclear_permutation_invariance_supp}) and (\ref{eq:joint_rotation_invariance_supp}).
\end{theorem*}

\begin{proof}
    Let us first consider permutation invariance, i.e.~Equation (\ref{eq:nuclear_permutation_invariance_supp}).  Let $r$ be produced by solving the flow 
    \begin{equation}
        dv/dt = \Gamma_t(v; R, Z) \text{ with } v(0) = z \sim \rho_z(\cdot) \text{ and } r = v(1)
    \end{equation}
    Consider a permutation $\pi$ on the nuclei, and let $\tilde{r}$ be the resulting electronic positions.  Then $\tilde{r}$ is produced by solving the flow 
    \begin{equation}
        d\tilde{v}/dt = \Gamma_t(\tilde{v}; \pi R, \pi Z) \text{ with } \tilde{v}(0) = z \sim \rho_z(\cdot) \text{ and }  \tilde{r} = \tilde{v}(1)
    \end{equation}
    However, we know that $\Gamma_t(\tilde{v}; \pi R, \pi Z) = \Gamma_t(\tilde{v}; R, Z)$.  Thus, $\tilde{r}$ is given by 
    \begin{equation}
        d\tilde{v}/dt = \Gamma_t(\tilde{v}; R, Z) \text{ with } \tilde{v}(0) = z \sim \rho_z(\cdot) \text{ and }  \tilde{r} = \tilde{v}(1)
    \end{equation}
    which is precisely equivalent to the equation for $r$; thus $\tilde{r} = r$, i.e.~the random variables representing the electronic positions are identical in both cases.  Thus, their distributions must be equal: $\rho(r; \pi R, \pi Z) = \rho(r; R, Z)$, so Equation (\ref{eq:nuclear_permutation_invariance_supp}) is established.

    Let us now turn to joint rotation invariance, i.e.~Equation (\ref{eq:joint_rotation_invariance_supp}).  As we know that $\Gamma_t$ satisfies rotation equivariance, i.e.~$\Gamma_t(\Theta v; \Theta R, Z) = \Theta \Gamma_t(v; R, Z)$, we may apply Theorems 1 and 2 from \citep{kohler2020equivariant} (noting that $R$ is irrelevant for the flow, which is entirely in $v$).  This yields immediately that $\rho(\Theta r; \Theta R, Z) = \rho(r; R, Z)$, so Equation (\ref{eq:joint_rotation_invariance_supp}) is established.
\end{proof}

\section{Proof of Theorem \ref{thm:induction_transformation}}
\label{prf:induction_transformation}

\begin{theorem*}
    Let $\phi_t(v; R, Z)$ be a function which is $\G$-equivariant with respect to $v$ i.e.~$\phi_t(g v; R, Z) = g \phi_t(v; R, Z)$ for $g \in \G$.  Let $\omega_t(v; R, z)$ be a function whose output is itself a rotation, i.e. $\omega_t(v; R, z) \in O(D)$.  Let $\omega_t$ be $\G$-invariant with respect to $v$, and $O(D)$-equivariant jointly with respect to $v$ and $R$ i.e.~$\omega_t(\Theta v; \Theta R, Z) = \Theta \omega_t(v; R, Z)$.  Finally, let both $\phi_t$ and $\omega_t$ be permutation-invariant jointly with respect to $R$ and $Z$ i.e.~$\phi_t(v; \pi R, \pi Z) = \phi_t(v; R, Z)$ and likewise for $\omega_t$. 
    Then the function 
    \begin{equation*}
        \Gamma_t(v; R, Z) = \zeta \phi_t(\zeta^{-1} v; \zeta^{-1} R, Z)
        \qquad \text{where} \qquad
        \zeta = \omega_t(v; R, Z)
    \end{equation*}
    satisfies the properties in Equation (\ref{eq:gamma_extra_properties_supp}) and is $\G$-equivariant with respect to $v$.
\end{theorem*}

\begin{proof}
    Let us begin with the first condition in Equation (\ref{eq:gamma_extra_properties_supp}), namely we wish to show that $\Gamma_t(v; \pi R, \pi Z) = \Gamma_t(v; R, Z)$.  Use tilde's to denote the variables after the permutation $\pi$ has been applied.  Thus,
    \begin{equation}
        \tilde{\zeta} = \omega_t(v; \pi R, \pi Z) = \omega_t(v; R, Z) = \zeta
    \end{equation}
    where we have used the fact that $\omega_t$ is permutation-invariant jointly with respect to $R$ and $Z$.  Then
    \begin{align}
        \Gamma_t(v; \pi R, \pi Z) 
        & = \tilde{\zeta} \phi_t(\tilde{\zeta}^{-1} v; \tilde{\zeta}^{-1} \pi R, \pi Z) \notag \\
        & = \zeta \phi_t(\zeta^{-1} v; \zeta^{-1} \pi R, \pi Z) \notag \\
        & = \zeta \phi_t(\zeta^{-1} v; \pi \zeta^{-1} R, \pi Z) \notag \\
        & = \zeta \phi_t(\zeta^{-1} v; \zeta^{-1} R, Z) \notag \\
        & = \Gamma_t(v; R, Z) 
    \end{align}
    where in the second line we have used the fact that $\tilde{\zeta} = \zeta$; in the third line, the fact that the operation of applying an identical rotation to a list of vectors commutes with  a permutation applied to that list of vectors; and in the fourth line, the fact that $\phi_t$ is permutation-invariant jointly with respect to $R$ and $Z$.  We have thus established the first condition in Equation (\ref{eq:gamma_extra_properties_supp}).

    Now let us turn to the second condition in Equation (\ref{eq:gamma_extra_properties_supp}), that is we need to show that $\Gamma_t(\Theta v; \Theta R, Z) = \Theta \Gamma_t(v; R, Z)$.  We have that
    \begin{equation}
        \tilde{\zeta} = \omega_t(\Theta v; \Theta R, Z) = \Theta \omega_t(v; R, Z) = \Theta \zeta
    \end{equation}
    where we have used the fact that $\omega_t$ is $O(D)$-equivariant jointly with respect to $v$ and $R$.  Then
    \begin{align}
        \Gamma_t(\Theta v; \Theta R, Z)
        & = \tilde{\zeta} \phi_t(\tilde{\zeta}^{-1} \Theta v; \tilde{\zeta}^{-1} \Theta R, Z) \notag \\
        & = \Theta \zeta \phi_t(\zeta^{-1}\Theta^{-1} \Theta v; \zeta^{-1}\Theta^{-1} \Theta R, Z) \notag \\
        & = \Theta \zeta \phi_t(\zeta^{-1} v; \zeta^{-1} R, Z) \notag \\
        & = \Theta \Gamma_t(v; R, Z)
    \end{align}
    as desired.

    Finally, let us turn to demonstrating the $\G$-equivariance of $\Gamma_t$ with respect to $v$.  Let $g \in \G$; then we have that
    \begin{equation}
        \tilde{\zeta} = \omega_t(g v; R, Z) = \omega_t(v; R, Z) = \zeta
    \end{equation}
    where we have used the fact that $\omega_t$ is $\G$-invariant with respect to $v$.  Then
    \begin{align}
        \Gamma_t(g v; R, Z)
        & = \tilde{\zeta} \phi_t(\tilde{\zeta}^{-1} g v; \tilde{\zeta}^{-1} R, Z) \notag \\
        & = \zeta \phi_t(\zeta^{-1} g v; \zeta^{-1} R, Z) \notag \\
        & = \zeta \phi_t(g \zeta^{-1} v; \zeta^{-1} R, Z) \notag \\
        & = \zeta g \phi_t(\zeta^{-1} v; \zeta^{-1} R, Z) \notag \\
        & = g \zeta \phi_t(\zeta^{-1} v; \zeta^{-1} R, Z) \notag \\
        & = g \Gamma_t(v; R, Z)
    \end{align}
    where in the second line we have used the fact that $\tilde{\zeta} = \zeta$; in the third and fifth lines, the fact that the operation of applying an identical rotation to a list of vectors commutes with  a permutation applied to that list of vectors; and in the fourth line, the fact that $\phi_t$ is $\G$-equivariant with respect to $v$.  This completes the proof.
\end{proof}

\section{Implementation of Continuous Normalizing Flow for Multiple Molecules}
\label{app:continuous_multiple}

We must implement both networks mentioned in Theorem 12: 
the functions $\phi_t$ and $\omega_t$.  The function $\phi_t$ is $\G$-equivariant, so that we may use the general recipe described in Appendix \ref{app:implementation}; however, it has the additional properties it depends on both $R$ and $Z$, and must be permutation-invariant jointly with respect to these two variables.  Therefore, the following minor modification may be made to the recipe described in Appendix \ref{app:implementation} (noting that the notation changes slightly as we no longer have layers $\l$ - the flow is continuous; and that we replace the variables $\gail$ with $v_{\alpha, i}$).
We compute a Deep Set \citep{zaheer2017deep} function on $R, Z$, i.e.~on the inputs $\{(R_I, Z_I)\}$; the output of this function is permutation-invariant by construction.  This output is then fed into the Fully Connected Layer with Spin Mixing as an extra input.  An alternative to the Deep Set approach is to apply a transformer to $R, Z$, where each token is the pair $(R_I, Z_I)$, and then apply an averaging step at the end; this will also produce a permutation-invariant function.

In order to implement the function $\omega_t$, recall that its output is a rotation matrix. Furthermore, it is $\G$-invariant in $v$; $O(D)$-equivariant with respect to $v$ and $R$ jointly; and permutation-invariant with respect to $R$ and $Z$ jointly.  We may use an EGNN architecture \citep{satorras2021n} jointly on electrons and nuclei.  In the EGNN:
\begin{itemize}
    \item The positions of the electrons and nuclei are initialized as $v$ and $R$ respectively.
    \item The hidden vectors of the electrons and nuclei are initialized in order to encode two things:
    \begin{enumerate}
        \item Whether the vertex corresponds to an electron or a nucleus.
        \item Properties of the vertex: (a) in the case of an electron, whether the spin is up or down; (b) in the case of a nucleus, the atomic number $Z_I$.
    \end{enumerate}
    This encoding can be achieved via combining one-hot vectors with linear projections of varying dimensionalities.
\end{itemize}
For each of the $D$ final layers of the EGNN, one may then take the position vectors for that layer and form an average over them; this yields a total of $D$ new vectors.  These $D$ vectors are clearly $\G$-invariant in $v$, as reordering within spins does not matter; permutation-invariant in $R$ and $Z$ jointly; and $O(D)$-equivariant with respect to $v$ and $R$ jointly, by the built-in equivariance properties of EGNNs.  We then take these $D$ vectors, and perform Gram-Schmidt on them to obtain a rotation matrix $\Theta$, noting that Gram-Schmidt retains the equivariance property.  A similar idea is discussed in \citep{kaba2023equivariance}.  This completes the implementation.


\end{document}